\documentclass[10pt,orivec]{llncs}
\usepackage{amssymb}
\usepackage{stmaryrd}
\usepackage{amsmath}
\usepackage{color}
\usepackage{verbatim}
\usepackage{cite}
\usepackage{algorithmic}
\usepackage{algorithm}
\usepackage{multirow}
\usepackage{hyperref}
\usepackage{quantikz}
\hypersetup{
    colorlinks=true,
    linkcolor=blue,
    filecolor=blue,      
    urlcolor=blue,
    citecolor=blue,
}

\renewcommand{\epsilon}{\varepsilon}
\newcommand{\functiondot}{\,\mathord{\cdot}\,}
\newcommand{\bigO}[1]{\mathcal{O}(#1)}

\newcommand{\setnocond}[1]{\{#1\}}
\newcommand{\setcond}[2]{\{\,#1 : #2\,\}}

\newcommand{\norm}[1]{\|#1\|}
\newcommand{\opnorm}[1]{\norm{#1}_{o}}

\newcommand{\radius}[1]{\varrho(#1)}
\newcommand{\maxmag}[1]{\kappa(#1)}

\newcommand{\epsneigh}[1][\epsilon]{\mathcal{N}_{#1}}
\DeclareMathOperator{\modulo}{mod}
\DeclareMathOperator{\lcm}{lcm}
\DeclareMathOperator{\id}{id}

\newcommand{\naturals}{\mathbb{N}}

\renewcommand{\dh}{\ensuremath{\mathcal{D(H)}}}

\newcommand{\lh}{\ensuremath{\mathcal{L(H)}}}
\newcommand{\ap}{AP}

\newcommand{\e}{\ensuremath{\mathcal{E}}}
\newcommand{\f}{\ensuremath{\mathcal{F}}}
\newcommand{\g}{\ensuremath{\mathcal{G}}}
\newcommand{\h}{\ensuremath{\mathcal{H}}}
\newcommand{\ii}{\ensuremath{\mathcal{I}}}
\renewcommand{\l}{\ensuremath{\mathcal{L}}}
\newcommand{\p}{\ensuremath{\mathcal{P}}}
\newcommand{\q}{\ensuremath{\mathcal{Q}}}
\newcommand{\x}{\ensuremath{\mathcal{X}}}
\newcommand{\m}{\ensuremath{\mathcal{M}}}
\renewcommand{\o}{\ensuremath{\mathcal{O}}}
\newcommand{\s}{S}
\newcommand{\bbh}{\ensuremath{\mathcal{L(\lh)}}}

\newcommand{\trj}{\sigma}
 
\newcommand{\vp}{\varphi}
\newcommand{\coloneqq}{:=}
\newcommand{\Coloneqq}{::=}
\newcommand{\tr}{{\rm tr}} 
 
\newcommand{\supp}[1]{\mathit{supp}( #1 )}
\DeclareMathOperator{\sspan}{span}
\newcommand{\spec}{\mathit{spec}}

\renewcommand{\bra}[1]{\langle #1 |}
\renewcommand{\ket}[1]{| #1 \rangle}
\renewcommand{\braket}[2]{\langle #1 | #2 \rangle}
\newcommand{\ketbra}[2]{ | #1 \rangle \langle #2 |}
\newcommand{\pair}[2]{\langle #1 , #2 \rangle}

\newcommand{\trjst}{\sigma}

\newcommand{\lfst}{L}

\newcommand{\mdst}{\models}

\renewcommand{\orcidID}[1]{\smash{\href{http://orcid.org/#1}{\protect\raisebox{-1.25pt}{\protect\includegraphics{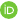}}}}}
\newcommand{\mvs}{\fontfamily{mvs}\fontencoding{U}\fontseries{m}\fontshape{n}\selectfont}
\newcommand\Letter{{\mvs\char66}}

\title{Measurement-based Verification of Quantum Markov Chains}

\titlerunning{Measurement-based Model Checking Quantum Statistical Systems}
\author{
Ji Guan\inst{1}\textsuperscript{(\Letter)}\orcidID{0000-0002-3490-0029}, 
Yuan Feng\inst{2}\orcidID{0000-0002-3097-3896}, 
Andrea Turrini\inst{1,3}\textsuperscript{(\Letter)}
\orcidID{0000-0003-4343-9323}
\and  
Mingsheng Ying\inst{4}\orcidID{0000-0003-4847-702X}
}
\institute{
Key Laboratory of System Software (Chinese Academy of Sciences) and State Key Laboratory of Computer Science, Institute of Software, Chinese Academy of Sciences, Beijing 100190, China \\ 
\texttt{\{guanj,turrini\}@ios.ac.cn}
\and 
Department of Computer Science and Technology, Tsinghua University, Beijing 100084, China
\and
Institute of Intelligent Software, Guangzhou, Guangzhou 511458, China
\and The Centre for Quantum Software and Information,  University of Technology Sydney, NSW 2007, Australia.
}

\begin{document}
\maketitle
\pagestyle{plain}
\begin{abstract}
Model-checking techniques have been extended to analyze quantum programs and communication protocols represented as quantum Markov chains, an extension of classical Markov chains. To specify qualitative temporal properties, a subspace-based quantum temporal logic is used, which is built on Birkhoff–von Neumann atomic propositions. These propositions determine whether a quantum state is within a subspace of the entire state space. In this paper, we propose the measurement-based linear-time temporal logic MLTL to check quantitative properties. MLTL builds upon classical linear-time temporal logic (LTL) but introduces quantum atomic propositions that reason about the probability distribution after measuring a quantum state. To facilitate verification, we extend the symbolic dynamics-based techniques for stochastic matrices described by Agrawal et al. (JACM 2015) to handle more general quantum linear operators (super-operators) through eigenvalue analysis. This extension enables the development of an efficient algorithm for approximately model checking a quantum Markov chain against an MLTL formula. To demonstrate the utility of our model-checking algorithm, we use it to simultaneously verify linear-time properties of both quantum and classical random walks. Through this verification, we confirm the previously established advantages discovered by Ambainis et al. (STOC 2001) of quantum walks over classical random walks and discover new phenomena unique to quantum walks.
\end{abstract}

\section{Introduction}
Model checking is a formal verification technique that is used to ensure the correctness of a system based on a given specification~\cite{baier2008principles}. In recent years, model checking has been applied to quantum systems and has become a powerful tool for verifying the behaviors and properties of quantum programs or communication protocols~\cite{ying2019model,ying2021model}. Similar to the classical case, the main components of model checking quantum systems are the system model and the temporal logic, which are used to mathematically describe the evolution of the system and specify its temporal properties, respectively.

\textbf{System Model.} Quantum Markov Chains (QMCs), which are the quantum extension of classical \emph{Markov chains} (MCs), provide an exceptional paradigm for modeling the evolution of quantum systems in various scenarios, including quantum control~\cite{Ticozzi2008a}, quantum information theory~\cite{guan2018structure}, quantum programming~\cite{ying2016foundations}, and quantum communication systems~\cite{wolf2012quantum}. Notably, quantum walks, which are a special class of QMCs, have been successfully employed in the design of quantum algorithms (for a survey of this research line, see~\cite{ambainis2003quantum,kempe2003quantum}). A QMC $\q$ is defined as a triple $\q = (\h, \e, \rho_{0})$ that corresponds to a classical Markov chain $(S, P, s_0)$, where $\h$ is a finite-dimensional Hilbert (linear) state space instead of the finite state set $S$, $\e$ is a \emph{super-operator} on $\h$ instead of the transition stochastic matrix $P$ on $S$, and $\rho_{0}$, which is a \emph{density matrix}, represents the initial state instead of $s_0$. 
Intuitively, the super-operator $\e(\functiondot)$, which is a linear mapping, models the dynamics of the system and transforms a state (density matrix) $\rho$ into another state $\e(\rho)$. 
Some special cases of QMCs have emerged, such as open quantum walks~\cite{Attal2012} and classical-quantum (super-operator valued) Markov chains~\cite{feng2013model}, where the latter resemble classical Markov chains but with the transition stochastic matrix $P = \setnocond{p_{i,j}}_{i,j \in S}$ being replaced by a transition set $\setnocond{\e_{i,j}}_{i,j \in S}$ of super-operators, while still maintaining the finite state set $S$.

\textbf{Temporal Logics.} The dynamic extension of Birkhoff–von Neumann quantum logic~\cite{birkhoff1936logic} was proposed to specify a wide range of temporal properties of quantum systems. In this approach, atomic propositions are used to describe qualitative properties of a quantum system, represented as closed subspaces of the system's state Hilbert space (whether or not a quantum state $\rho$ is in a subspace $\x$ of $\h$). Furthermore, QMCs are abstracted as subspace transition systems, and the temporal properties of interest are represented by infinite sequences of sets of atomic propositions (subspaces)~\cite{ying2021model}. This subspace-based temporal logic allows for the specification of linear-time properties, such as invariants and safety properties, for quantum automata (a simplified form of QMCs)~\cite{moore2000quantum}. Model-checking algorithms have been developed for the subspace-based temporal logic in~\cite{ying2014model}, and the (un)decidability of model checking quantum automata has been studied in~\cite{li2014decidable}. However, a limitation of model checking QMCs against subspace-based temporal logics is that it can only handle qualitative properties, which means that only simple examples can be checked. Given that the power of quantum systems lies in their probabilistic nature, it is crucial to be able to check probabilistic (quantitative) properties.

To address this issue, in this paper we propose a measurement-based linear-time temporal logic to capture these properties and develop a model-checking algorithm to check the quantitative properties of QMCs.
More specifically, we observe that the properties of the quantum systems in question can be described using quantum measurements, which extract classical (probabilistic) information from quantum states. Building on this observation, we introduce measurement-based atomic propositions to describe static quantitative properties, namely, the measurement outcome probability of a quantum state under a measurement. This can be seen as a generalization of the subspace-based atomic propositions of the Birkhoff–von Neumann quantum logic~\cite{birkhoff1936logic}, where a quantum state $\rho$ in a subspace $\x \subseteq \h$ can be regarded as having a measurement outcome probability of $1$ under the measurement onto $\x$. By combining with standard linear-time temporal logic (LTL)~\cite{Pnueli77LTL,baier2008principles}, we obtain MLTL, the measurement-based LTL, to specify the temporal quantitative properties of quantum systems.

In order to develop an algorithm for model checking QMCs against MLTL formulas, we extend the Thiagarajan's approximate verification~\cite{agrawal2015approximate} of a stochastic transition matrix $P$ to encompass more general linear operators in quantum systems, e.g., the super-operators. The key technique for this generalization is based on the eigenvalue analysis of QMCs, which simplifies the previous work based on the \emph{Bottom Strongly Connected Component (BSCC) decomposition}~\cite{baier2008principles} of the state space of classical Markov chains~\cite{agrawal2015approximate}. Subsequently, we provide an effective procedure for the approximate model checking of QMCs against MLTL formulas. In Section~\ref{sec:experiments}, we provide several case studies to illustrate how our model and algorithm can be applied in a quantum walk. 
 These case studies help to verify the previously established advantages of quantum walks over classical random walks, as discovered by Ambainis et al.~\cite{ambainis2001one}. Additionally, we explore new phenomena unique to quantum walks when we verify the same MLTL formulas on both types of walks.

In summary, this paper makes the following main contributions:
\begin{enumerate}
    \item \emph{Introducing} a quantum temporal logic, called measurement-based linear-time logic (MLTL), which allows for specifying quantitative properties of QMCs.
    \item \emph{Generalizing} symbolic dynamics-based verification techniques of the transition stochastic matrix given in~\cite{agrawal2015approximate} to more general quantum linear operators (super-operators) by \emph{eigenvalue analysis}; based on this, a model-checking algorithm for QMCs against MLTL is developed.
    \item  \emph{Verifying} numerous quantitative properties of quantum walks through our model-checking algorithm as case studies.  This serves to validate the established advantages of quantum walks over their classical counterparts and discover new phenomena unique to quantum walks.
\end{enumerate}

\subsection{Related Works and Challenges}
To provide a suitable context for our work, let us delve deeper into the discussion of related works and the challenges we encounter in this paper.

\textbf{Hybrid vs.\@ Quantum Temporal Logic.} In a previous study~\cite{feng2013model}, a specialized type of quantum Markov chain known as super-operator-valued quantum Markov chain was proposed. This model was designed for the purpose of modeling quantum programs and quantum cryptographic protocols. Additionally, a quantum extension of the probabilistic computation tree logic (PCTL) called quantum computation tree logic (QCTL) was introduced, along with a model-checking algorithm specifically tailored for this Markov model.

However, these hybrid temporal logic approaches heavily rely on the classical state graph and are not applicable to quantum systems. This is because quantum systems have a continuous state space and an infinite number of states, making it impossible to obtain a connection graph. In order to address this limitation and specify the properties of quantum Markov chains, we propose a measurement-based linear-time temporal logic (MLTL) approach that does not require a connection graph. This allows us to directly reason about the transitions of measurement outcome probability distributions. Our MLTL approach adopts classical LTL with quantum physical interpretation, providing the advantage of utilizing existing classical techniques and facilitating contributions from the classical model-checking community to the field of quantum computing.

\textbf{Classical vs.\@ Quantum Markov Chains.} The main technique used for model checking classical Markov chains in~\cite{agrawal2015approximate} involves studying the periodicity of states in sub-chains obtained through BSCC decomposition~\cite{baier2008principles}, a widely-used method in the field of model checking Markov chains. However, when it comes to the quantum extension of this decomposition, as described in~\cite{ying2013reachability}, the BSCC decomposition of QMCs is not unique but there are infinitely many decompositions due to the continuous state space, unlike the classical case. Consequently, we cannot directly generalize Thiagarajan's \emph{approximate verification}~\cite{agrawal2015approximate}. To overcome this unique challenge posed by quantum mechanics, we propose a method based on eigenvalue analysis to directly explore the periodic properties of QMCs. As a result, QMCs are not always periodically stable like classical Markov chains. We provide a way to determine the periodic stability depending on the initial states of QMCs. With these efforts, we successfully extend the idea of approximate model checking to work for QMCs, pushing the application boundary of such model-checking techniques to more general linear systems.

\textbf{Model Checking Quantum Walks.} 
As a result of the phenomenon known as \emph{quantum interference}~\cite{venegas2012quantum}, quantum walks can propagate at significantly faster or slower rates compared to their classical counterparts~\cite{ambainis2001one}. Quantum walks have gained attention due to their potential application in the development of randomized algorithms, with various quantum algorithms incorporating this concept~\cite{ambainis2003quantum,kempe2003quantum}. Notably, in certain search problems, quantum walks can offer quadratic or exponential speedups compared to classical algorithms.
    
Previously, researchers have conducted case-by-case studies on the dynamic properties of quantum walks, requiring the introduction of various techniques depending on the specific property and underlying topological structure of the walks. In this paper, we introduce a model-checking method that overcomes this limitation by allowing for the automatic verification of a wide range of properties of the walks. To model quantum walks, we utilize QMCs, and to specify the relevant dynamical properties, we employ MLTL formulas. By simultaneously verifying classical and quantum walks using our model-checking algorithm, we can confirm the advantages of quantum walks that have already been established in~\cite{ambainis2001one}, as well as discover new phenomena that are distinct from classical random walks. We anticipate that these new phenomena, discovered by our model-checking algorithm, will contribute to the development of more efficient quantum walk-based algorithms with enhanced speedup capabilities.

\section{Preliminaries}\label{Sec:preliminaries}
In this section, we aim to explain the three components that appear in a QMC $\q = (\h, \e, \rho_0)$, as well as quantum measurements, an essential part of our MLTL.

\textbf{Quantum State Space $\h$.} The \emph{state space} of a quantum system is a finite-dimensional linear space $\h$, which is commonly known as the \emph{Hilbert space} in the field of quantum computing. A \emph{quantum pure state} is represented by a unit complex column vector $\psi$ in $\h$. In the field of quantum computing, the \emph{bra-ket notation} is widely used to represent quantum states, making it easier to perform calculations that frequently arise in quantum mechanics. This notation uses angle brackets, $\langle$ and $\rangle$, along with a vertical bar $|$, to construct ``bras'' and ``kets'' that represent row and column vectors, respectively. The following list provides the notation used in this paper to represent linear algebra concepts:
\begin{enumerate}
    \item $\ket{\psi}$ represents a unit complex column vector (quantum pure state) in $\h$, labeled with $\psi$;
    \item $\bra{\psi}\coloneqq \ket{\psi}^\dagger$ denotes the complex conjugate and transpose of $\ket{\psi}$;
    \item $\braket{\psi_1}{\psi_2} \coloneqq \bra{\psi_1} \ket{\psi_2}$ represents the inner product of $\ket{\psi_1}$ and $\ket{\psi_2}$;
    \item $\ketbra{\psi_1}{\psi_2} \coloneqq \ket{\psi_1}\cdot\bra{\psi_2}$  denotes the outer product of $\ket{\psi_1}$ and $\ket{\psi_2}$.
\end{enumerate}

It is important to note that any vector in $\h$ can be linearly represented by a \emph{computational basis}, which is a set of mutually orthogonal unit vectors. In order to compare with classical Markov chains denoted as $(S, P, \mu_0)$, we use the finite state set $S = \setnocond{s_0, \dotsc, s_{d-1}}$ to label the computational basis of $\h$, with dimension $d$, as $\setnocond{\ket{s_0}, \dotsc, \ket{s_{d-1}}}$. Here, $\ket{s_k}$ is a unit column vector with the $k$-th element being $1$ and the remaining elements being $0$ (the index starts from $0$). Then $\h$ is denoted as $\h = \sspan\setnocond{\ket{s_0}, \dotsc, \ket{s_{d-1}}}$.

Using this basis, any quantum state $\ket{\psi}$ in $\h$ can be expressed as a linear combination of $\setnocond{\ket{s_0}, \dotsc, \ket{s_{d-1}}}$ with complex coefficients $a_k$: $\ket{\psi} = \sum_{k=0}^{d-1} a_k \ket{s_k}$ with the normalization condition $\braket{\psi}{\psi} = \sum_{k=0}^{d-1} a_{k} a_{k}^*=1$, where $a_k^*$ is the complex conjugate of $a_k$.  In the case of a $2$-dimensional space, we have:
\[\ket{s_0}=\left(\begin{matrix}
    1\\
    0
\end{matrix}\right )\qquad \ket{s_1}=\left(\begin{matrix}
    0\\
    1
\end{matrix}\right )\qquad\ket{\psi}=\left(\begin{matrix}
    a_0\\
    a_1
\end{matrix}\right )\qquad\bra{\psi}=\left(\begin{matrix}
    a_0^*, a_1^*
\end{matrix}\right ).\]
It is evident that a quantum pure state has the capability to depict a \emph{superposition} of state set $S=\{s_0,\ldots,s_{d-1}\}$ as $\ket{\psi}=\sum_{k=0}^{d-1}a_k\ket{s_k}$. The superposition is a unique feature of quantum systems and the main reason for the advantages of quantum algorithms over their classical counterparts~\cite{nielsen2010quantum}.

\textbf{Quantum Mixed State $\rho$.} In quantum mechanics, uncertainty is a common characteristic of quantum systems, arising from quantum noise and measurements. To describe the uncertainty of possible quantum pure states, the concept of \emph{quantum mixed state} $\rho$ on $\h$ is introduced. It can be represented as
\begin{equation}\label{Eq:ensemble}
    \begin{aligned}
        \rho = \sum_{k} p_{k}\ket{\psi_{k}}\bra{\psi_{k}}.
    \end{aligned}
\end{equation}
Here, $\setnocond{(p_{k}, \ket{\psi_{k}})}_{k}$ represents an ensemble, indicating that the quantum state is at $\ket{\psi_{k}}$ with probability $p_{k}$. This concept can also be used to describe the uncertainty of a classical probability distribution, where each $\ket{\psi_k}$ represents $s_k$ by $\ket{\psi_k} = \ket{s_k}$. Specifically, a (row) probability distribution $\mu = (p_0, \dotsc, p_{d-1})$ over the state set $S = \setnocond{s_0, \dotsc, s_{d-1}}$ can be represented by a quantum mixed state. This representation involves a diagonal matrix on $\h$, where the diagonal elements correspond to the probabilities $p_k$ as follows.
\begin{equation}\label{Eq:encode}
    \begin{aligned}
        \rho_{\mu} = \sum_{k} p_{k}\ket{s_{k}}\bra{s_{k}}=\mathrm{diag}(p_0,\ldots,p_{d-1}).
    \end{aligned}
\end{equation}
Hence, a quantum mixed state $\rho$ is an extension of a probability distribution $\mu$. Generally, the quantum uncertainty is more complex because the ensemble decomposition in Eq.~\eqref{Eq:ensemble} of a quantum state $\rho$ can have infinitely many variants. This means that $\rho$ can represent multiple ensembles $\setnocond{(p_{k},\ket{\psi_{k}})}_{k}$ simultaneously.

From a mathematical perspective, a quantum mixed state $\rho \in \lh$ is a linear operator ($d$-by-$d$ matrix) on $\h$ that satisfies three conditions: 1) \emph{Hermitian} $\rho^\dagger = \rho$; 2) \emph{positive semi-definite} $\bra{\psi} \rho \ket{\psi} \geq 0$ for all $\ket{\psi} \in \h$; and 3) \emph{unit trace} $\tr(\rho) = \sum_{k} \bra{s_k} \rho \ket{s_k} = 1$, where $\tr(\rho)$ is the trace of $\rho$ and represents the sum of the diagonal elements of $\rho$. Here, $\lh$ denotes the set of linear operators on $\h$. Let $\dh \subseteq \lh$ be the set of all quantum mixed states on $\h$. To avoid any ambiguity, in the subsequent discussion, the term ``quantum states'' will specifically refer to quantum mixed states, given that we are considering the broader scenario.

\textbf{Quantum Evolution $\e$.} In the realm of quantum computing, the evolution of a quantum system is commonly represented by the equation
\begin{equation}\label{Eq:evolution}
	\rho' = \e(\rho).
\end{equation}
Here, $\e$ is referred to as a \emph{super-operator}. Mathematically, $\e(\functiondot)$ is a linear mapping from $\lh$ to $\lh$, allowing for the transformation of one quantum state $\rho$ into another $\rho'$. As stated by the \emph{Kraus representation theorem}~\cite{choi1975completely}, $\e$ can be characterized by a finite set of $d$-by-$d$ matrices $\setcond{E_{k}}{0 \leq k \leq m-1} \subseteq \lh$, where $m \leq d^{2}$. The expression is given by $\e(\rho) = \sum_{k=0}^{m-1} E_{k} \rho E_{k}^{\dag}$ for all $\rho\in \dh$.

This representation also satisfies the trace-preserving condition $\sum_{k} E_{k}^{\dag} E_{k} = I$, where $I$ is the identity matrix on $\h$ and  $\dag$ denotes the complex conjugate and transpose of matrices. In other words, for all $\rho \in \dh$, we have $\tr(\e(\rho)) = \tr(\rho)$ by
$ \tr(\e(\rho)) = \tr(\sum_{k}E_k\rho E_k^\dagger)=\tr(\sum_{k}E_k^\dagger E_k\rho)=\tr(\rho).$

An example of the use of super-operators can be found in Section~\ref{Sec:quantum_walk}, where a super-operator is employed to represent the evolution of quantum walks. In the degenerate scenario, the Kraus operator is simplified to only include a \emph{unitary matrix} $U$ on $\h$ (where $U^\dagger U=U^\dagger U=I$), and $\e(\rho)=U\rho U^\dagger$.

\textbf{Quantum Measurement.} To extract information from a quantum state, a \emph{quantum measurement} is performed. This measurement yields a classical outcome which is represented as a probability distribution over the possible results. Mathematically, a quantum measurement is described by a set $\setnocond{M_{k}}_{k \in \o}$ of positive semi-definite matrices on the state (Hilbert) space $\h$, where $\o$ is a finite set of possible outcomes. The measurement process is probabilistic: if the quantum system is in a state $\rho$ before the measurement, the probability of obtaining the outcome $k$ is given as follows.
\[p_{k} = \tr(M_{k} \rho).\]
Note that the measurement $\setnocond{M_{k}}_{k \in \o}$ satisfies the \emph{unity condition} $\sum_{k} M_{k}=I$, which guarantees that the total probability of all outcomes is equal to $1$. In other words, $\sum_{k} \tr(M_{k} \rho) = \tr(\sum_{k} M_{k} \rho) = \tr(\rho) = 1$.

In the case where we want to extract the classical probability distribution $\mu$ encoded in the quantum state $\rho_{\mu}$ as shown in Eq.~\eqref{Eq:encode}, we can choose the measurement $\setnocond{M_{k} = \ketbra{s_k}{s_k}}_{s_k \in S}$. In this scenario, the measurement probability of obtaining outcome $s_k$ is given by
\[
    \tr(\ketbra{s_k}{s_k} \rho_\mu) 
    = \sum_{l} \bra{s_l} \ketbra{s_k}{s_k} \rho_{\mu} \ket{s_l} 
    = \bra{s_k} \rho_{\mu} \ket{s_k} 
    = \bra{s_k}(\sum_{j} p_j \ketbra{s_j}{s_j}) \ket{s_k} 
    = p_k.
\]
The above equations rely on the mutual orthogonality of the computational basis $\setnocond{\ket{s_0}, \dotsc, \ket{s_{d-1}}}$, meaning that $\braket{s_j}{s_l} = 0$ for $j \neq l$, and also the fact that each $\ket{s_k}$ is normalized, represented by $\braket{s_k}{s_k}=1$.

It should be noted that after the measurement, the state will collapse or be altered, depending on the measurement outcome $k$, which distinguishes quantum computation from classical computation. For instance, in the case of a \emph{projection measurement} denoted as $\setnocond{P_{k}}_{k \in \o}$, the state after obtaining outcome $k$ is given by $P_{k}\rho P_{k}/p_k$ with $p_k=\tr(P_k\rho)$. Here, each positive semi-definite matrix $P_{k}$ represents a projection operator ($P_{k}^2=P_{k}$), and a specific example of a projection measurement is the above measurement $\setnocond{\ketbra{s_k}{s_k}}_{s_k \in S}$. Another concrete example is provided in Example~\ref{Exa:Quantum_walk} for the case of quantum walks. For other scenarios involving post-measurement states that are not encountered in this paper, please refer to~\cite[Section 2.2.3]{nielsen2010quantum}.

\section{Quantum Markov Chains}

In this section, we present the formal definition of QMCs. For a more detailed discussion, we refer the interested readers to~\cite{ying2021model}.

\begin{definition}\label{Def:QMC}
    A QMC is a tuple $\g = (\h, \e, \rho_{0})$, where $\h$ is a finite-dimensional Hilbert space, $\e$ is a super-operator on $\h$, and $\rho_{0} \in \dh$ is an initial state. 
\end{definition}

The execution of $\g$ is naturally described by the trajectory of quantum states:
\begin{equation}
\label{q-traj}
	\trjst(\g) \coloneqq \rho_{0}, \e(\rho_{0}), \e^{2}(\rho_{0}), \ldots.
\end{equation}

QMCs are a direct extension of classical Markov chains and can simulate their execution; see the following example. 

\begin{example}[Classical Markov chains as QMCs]\label{Exa:modeling_MC}
Not surprisingly, any classical Markov chain $(\s, P, \mu_{0})$ can be effectively encoded as a QMC.
We can use $\h = \sspan\setnocond{\ket{s_0}, \dotsc, \ket{s_{d-1}}}$ to encode $\s$, and $\e$ can be a super-operator with Kraus operators $\setnocond{E_{k,l} = \sqrt{p_{k,l}} \ketbra{s_{l}}{s_{k}}}_{s_k,s_l \in \s}$ that encode the probabilities $p_{k,l}$ of $P$.

It can be easily verified that $\e$ is a valid super-operator, i.e., $\sum_{k,l} E_{k,l}^\dagger E_{k,l} = I$. Furthermore, let $\rho_{0} = \sum_{s \in \s} \mu_{0}(s) \ket{s} \bra{s}$ encode the initial probability distribution $\mu_{0}$. Then, the QMC $(\h, \e, \rho_{0})$ can fully simulate the behavior of $(\s, P, \mu_{0})$ in the sense that for all $n \geq 0$:
\[
    \e^{n}(\rho_{0}) = \sum_{k=0}^{d-1} \mu_{n}(s_k) \ket{s_k} \bra{s_k} = (\mu_{n}(s_0), \dotsc, \mu_{n}(s_{d-1})).
\]
Here, $\mu_{n} = \mu_{0} P^{n}$, $\mu_{n}(s_k)$ represents the $k$-th entry of $\mu_{n}$, and $\e^{0} = \id_{\h}$, which is the \emph{identity super-operator} with only one Kraus operator $\setnocond{I}$. The proof follows a straightforward induction on $n$.
\end{example}

\subsection{Quantum Walks}\label{Sec:quantum_walk}
The case study of this paper is to explore the new and advanced properties of quantum walks compared to classical random walks by model checking QMCs. This exploration begins with modeling quantum walks using QMCs. Before presenting this, we introduce one-dimensional quantum walks with absorbing boundaries in the following example.   Furthermore, we also need the bra-ket notation  $\ket{\psi_1}\ket{\psi_2} \coloneqq \ket{\psi_1} \otimes \ket{\psi_2}\in \h_1 \otimes \h_2$, which represents the composition (tensor product) of $\ket{\psi_1}$ and $\ket{\psi_2}$ in the Hilbert spaces $\h_1$ and $\h_2$, respectively.

\begin{example}[Quantum walk with absorbing boundaries]\label{Exa:Quantum_walk}
We consider an unbiased quantum walk on a one-dimensional lattice indexed from $s_0$ to $s_d$, with the boundaries $s_0$ and $s_{d}$ being absorbing. 

\textbf{State Space.} Let $\h_{p} = \sspan\setnocond{\ket{s_0}, \dotsc, \ket{s_{d}}}$ be a $(d+1)$-dimensional Hilbert space, where the pure state (unit vector) $\ket{s_k}$ represents the position $s_k$ for each $0 \leq k \leq d$. 
In order to facilitate the evolution of the quantum walk, an additional coin space is required. Let $\h_{c}$ be the coin (direction) space, which is a $2$-dimensional Hilbert space with orthonormal basis states $\ket{L}$ and $\ket{R}$, indicating the directions left and right, respectively.
Therefore, the state space of the quantum walk is $\h = \h_{p} \otimes \h_{c}$. 
 
\textbf{Initial State.} The initial state is $\rho_{0}=\ketbra{\psi_0}{\psi_0}$ with $\ket{\psi_{0}} = \ket{s_k} \ket{X}$, indicating the initial position $s_k$ and direction $X$, where $0 \leq k \leq d$ and $X \in \setnocond{L, R}$.

\textbf{Evolution.} Each step of the walk consists of three operations:

First, measure the current position of the system to determine whether it is absorbing positions $s_0$ or $s_d$. If the position is $s_0$ or $s_d$, then the walk terminates. The measurement is described by
\[\setnocond{M_{\mathit{yes}} = (\ket{s_0}\bra{s_0} + \ket{s_n}\bra{s_n}) \otimes I_{c}, M_{\mathit{no}} = I -M_{\mathit{yes}}}.\]
Here, $I_{c}$ and $I$ are the identity operators on $\h_{c}$ and $\h$, respectively. According to the principles of measurements, from the current state $\rho$, the walk terminates at state $\rho_\mathit{yes}$ with probability $p_{\mathit{yes}}$, and it continues with state $\rho_{\mathit{no}}$, with probability $p_{\mathit{no}}$, where \[p_{\mathit{x}} = \tr( M_{\mathit{x}} \rho) \quad \text{ and }\quad
\rho_{\mathit{x}} = {M_{\mathit{x}} \rho M_{\mathit{x}}^\dagger}/p_{\mathit{x}} \quad \text{ for } \quad \mathit{x} = \mathit{yes}, \mathit{no}.\]

Second, apply an unbiased ``coin-tossing'' (unitary) operator on the coin space $\h_c$. This operator, denoted as $U_H$, is given by:
\[U_{H} = \frac{1}{\sqrt{2}}\left(\begin{matrix}
				1&1\\
				1&-1
			\end{matrix}\right) = \ket{+}\bra{L} + \ket{-}\bra{R}.\]
Here, $\ket{\pm} = \frac{1}{\sqrt{2}} (\ket{L} \pm \ket{R})$. The operator  $U_H$ represents the \emph{Hadamard operator} on $\h_{c}$, where the probabilities of going left and right are both equal to $0.5$.

Third, perform a shift (unitary) operator on the space $\h$. The shift operator, denoted as $U_{S}$, is given by:
\[U_S = \sum_{k=0}^{d} \ket{s_{k \ominus 1}} \bra{s_k} \otimes \ket{L} \bra{L} + \ket{s_{k \oplus 1} }\bra{s_k} \otimes \ket{R} \bra{R}.\]
The intuitive meaning of the operator $U_S$ is that the system walks one step left or right based on the coin state on $\h_c$. Here, $\oplus$ and $\ominus$ represent addition and subtraction modulo $d+1$, respectively.

In summary, in each step, given the current state $\rho$ as the input, the quantum walk transforms $\rho$ into $\rho' = U \rho_{no} U^\dagger$ with a probability $p_{no}$, where $U = U_S (I_{p} \otimes U_H)$ and $I_p$ is the identity operator on $\h_p$. Additionally, the walk terminates at state $\rho_{yes}$ with a probability $p_{yes}$. The resulting state $\rho'$ then serves as the input state for the subsequent step of the quantum walk.
\end{example}
For a better understanding, Appendix~\ref{appendix:quantum_walk} provides a visual representation (Fig.~\ref{fig:quantum_walks}) that showcases the evolution of the quantum walk in Example~\ref{Exa:Quantum_walk} from a physical perspective. 

Now, we demonstrate how to represent the quantum walk given in Example~\ref{Exa:Quantum_walk} using a QMC.
To begin, we introduce the super-operator $\e$, which incorporates the unitary operator $U$ representing the combined evolution of the second and third operations of the quantum walk. For any given $\rho \in \dh$, we let $\e(\rho)$ be
\[
\e(\rho) = U M_{\mathit{no}} \rho M_{\mathit{no}}^\dagger U^\dagger + M_{\mathit{yes}} \rho M_{\mathit{yes}}^\dagger.
\]
It is worth noting that $M_{\mathit{no}} M_{\mathit{yes}} = 0$ and $M_{\mathit{yes}}^2 = M_{\mathit{yes}}$, indicating that once the QMC terminates, its state remains unchanged. By using induction on the number of steps, we can easily verify that the evolution of the quantum walk can be modeled by the QMC $(\h_{p} \otimes \h_{c}, \e, \rho_{0} = \ket{\psi_{0}} \bra{\psi_{0}})$.

\section{Measurement-based Linear-time Temporal Logic}
\label{Sec:logic}

In this section, we introduce a specification language called \emph{Measurement-based Linear-time Temporal Logic} (MLTL) for describing the properties of quantum systems. MLTL is similar to ordinary LTL, but its atomic propositions are interpreted in the context of quantum computing. It also expands on the subspace-based atomic propositions introduced by Birkhoff and von Neumann~\cite{birkhoff1936logic}.

\subsection{Measurement-based Atomic Propositions}
\label{sec:MAP}
As mentioned in Section~\ref{Sec:preliminaries}, a quantum measurement is the process of extracting classical information from quantum states. A quantum measurement can be represented by a finite set $\setnocond{M_{k}}_{k \in \o}$ of positive semi-definite matrices with the unity condition $\sum_{k} M_k = I$. Each matrix $M_{k}$ is called a \emph{measurement operator} and it is associated with the probability $\tr(M_k \rho)$ of obtaining the outcome $k$ when measuring the quantum state $\rho$. Mathematically, a measurement operator $M$ is positive semi-definite and satisfies $M \leq I$, which means that $I - M$ is also positive semi-definite.

In the following discussion, our main focus is on the measurement operator $M_k$ and its associated probability $\tr(M_k \rho)$. Therefore, we will often disregard the specific outcome value $k$ and simply refer to the measurement operator as $M$. Additionally, we will not explicitly mention the measurement consisting of $M$, as a binary quantum measurement $\setnocond{M, I-M}$ can be determined based on $M$ itself (as seen in Example~\ref{Exa:Quantum_walk} with the measurement operators $M_{\mathit{yes}}$ and $M_{\mathit{no}}$).

Our atomic propositions are designed to estimate the probability of the outcome $\tr(M \rho)$ after measuring the quantum state $\rho$, given the measurement operator $M$.
\begin{definition}
\label{def-sat}
    Given a Hilbert space $\h$,
    \begin{enumerate}
	\item an atomic proposition in $\h$ is a pair $\pair{M}{\ii}$, where $M$ represents a measurement operator on $\h$ and $\ii \subseteq [0,1]$ is an interval; 
	\item a state $\rho \in \dh$ satisfies $\pair{M}{\ii}$, written $\rho \models \pair{M}{\ii}$, if the outcome probability of the measurement operator $M$ applied to $\rho$ falls within the interval $\ii$, that is, $\tr(M \rho) \in \ii$.
    \end{enumerate}
\end{definition}

In terms of the expressive power of our measurement-based atomic propositions, it is worth noting that they extend the existing atomic propositions in both the classical and quantum domains. 
\begin{description}
    \item[Classical:] the atomic proposition $\pair{M}{\ii}$ of $\rho$ expands upon the proposition $\pair{s_k}{\ii}$ of $\mu$ in interval linear-time temporal logic. This classical logic has been used to specify (static) properties of classical Markov chains~\cite{agrawal2015approximate} and continuous-time Markov chains~\cite{guan2022probabilistic}. The proposition asserts that the probability of state $s_k\in\s$ in distribution $\mu$ over $\s$ falls within the interval $\ii$. The extension is achieved by utilizing $\rho_{\mu} = \sum_{k} \mu(s_k) \ketbra{s_k}{s_k}$ in Eq.~\eqref{Eq:encode} and $M = \ketbra{s_k}{s_k}$, resulting in $\tr(M \rho_{\mu}) = \mu(s_k)$. Consequently, $\tr(M \rho_{\mu}) \in \ii$ if and only if $\mu(s_k) \in \ii$.

    \item[Quantum:] furthermore, an atomic proposition $\pair{M}{\ii}$ of a mixed state $\rho$ can encode the subspace-based proposition $\x \subseteq \h$ of a pure state $\ket{\psi}$ in the Birkhoff–von Neumann quantum logic. This proposition asserts that $\ket{\psi} \in \x$. This observation is made by setting $M = P_{\x}$, which represents the projection onto $\x$, $\ii = [1,1]$, and $\rho = \ketbra{\psi}{\psi}$. As a result, we can conclude that $\rho = \ketbra{\psi}{\psi} \models \pair{P_{\x}}{[1,1]}$ (i.e., $\tr(P_{\x} \rho) = 1$) is equivalent to $\ket{\psi} \in \x$.
\end{description}
To demonstrate the practicality of our atomic propositions, we will present a series of specific instances in Example~\ref{ex:qwc} later. These examples will effectively illustrate the characteristics and properties of quantum walks.

From an algorithmic perspective, we gather a finite number of pairs $\pair{M}{\ii}$ as the set of atomic propositions that are based on quantum measurements. This set is denoted as $\ap$.

\subsection{Quantum Linear-Time Temporal Logic}
\label{ssec:ltl}
In this section, we enhance the linear-time temporal logic~\cite{Pnueli77LTL,baier2008principles} by incorporating the newly introduced measurement-based atomic propositions. This will result in the formation of our measurement-based linear-time temporal logic MLTL.

The MLTL formulas, which involve the use of measurement-based $\ap$, are defined according to the following syntax:
\[
    \vp \Coloneqq  \mathbf{true} \mid a \mid \neg\vp \mid \vp_{1} \lor \vp_{2} \mid \bigcirc\vp \mid \vp_{1} U \vp_{2},
\]
where $a = \pair{M}{\ii} \in \ap$. 
We can also derive additional standard Boolean operators and temporal modalities such as $\lozenge$ (\emph{eventually}) and $\square$ (\emph{always}) using conventional methods.

The semantics of MLTL is also defined in a familiar manner. 
For any infinite word $\xi$ over $2^{\ap}$ and any MLTL formula $\vp$ over $\ap$, the satisfaction relation $\xi \models \vp$ is defined by induction on the structure of $\vp$:
\begin{itemize}
\item 
    $\xi \models \mathbf{true}$ always holds;
\item 
    $\xi \models a$ iff $a \in \xi[0]$;
\item 
    $\xi \models \neg\vp$ iff it is not the case that $\xi \models \vp$ (written $\xi \not\models \vp$);
\item 
    $\xi \models \vp_{1} \lor \vp_{2}$ iff $\xi \models \vp_{1}$ or $\xi \models \vp_{2}$;
\item 
    $\xi \models \bigcirc\vp$ iff $\xi[1+] \models \vp$; and
\item 
    $\xi \models \vp_{1} U \vp_{2}$ iff there exists $k \geq 0$ such that $\xi[k+] \models \vp_{2}$ and for each $0 \leq j < k$, $\xi[j+] \models \vp_{1}$,
\end{itemize}
where $\xi[k]$ and $\xi[k+]$ denote the $(k+1)$-th element and the $(k+1)$-th suffix of $\xi$, respectively. 
The indexes start from zero so that, say, $\xi = \xi[0+]$.  In addition, the \emph{semantics} $\l_\omega(\vp)$ of $\vp$ is defined as the language containing all infinite words over $2^{\ap}$ that satisfy $\vp$: $\l_\omega(\vp) = \setcond{\xi \in (2^{\ap})^{\omega}}{\xi \models \vp}.$

Now, let us extend the satisfaction relation $\rho \models a$ to $\g \models \vp$ between a QMC $\g$ and an MLTL formula $\vp$.  
To this end, we introduce the labeling function: 
\begin{equation}
\label{l1-fun}
	\lfst \colon \dh \to 2^{\ap},\qquad \lfst(\rho) = \setcond{a \in \ap}{\rho \models a}
\end{equation}
which assigns to each quantum state $\rho$ the set of atomic propositions in $\ap$ satisfied by $\rho$. 
We further extend the labeling function to sequences of quantum states by setting $\lfst(\rho_{0}, \rho_{1}, \ldots) = \lfst(\rho_{0}), \lfst(\rho_{1}), \ldots$ as usual. 
Then we define: 
\[
	\text{$\g \mdst \vp$  if and only if $\lfst(\trjst(\g))\in\mathcal{L}_\omega(\vp)$}
\] 
where $\trjst(\g)$ is the state trajectory of $\g$ as defined in Eq.~\eqref{q-traj}.

We now exhibit realistic settings where our approach leads to valuable insights for the quantum walk presented in Example~\ref{Exa:Quantum_walk}.

\begin{example}[Quantum walk with absorbing boundaries, continued]\label{ex:qwc}
	Given a finite set of intervals $\setnocond{\ii_l \subseteq [0,1]}_{l=0}^L$, let 
	\[
	\ap = \setcond{\pair{M_{s_k}}{\ii_l}}{M_{s_k} = \ket{s_k} \bra{s_k} \otimes I_{c}, 0 \leq k \leq d, 0 \leq l \leq L}
	\] with the atomic proposition $\pair{M_{s_k}}{\ii_l}$ asserting that $\tr(M_{s_k} \rho) \in \ii_l$ for $0 \leq k \leq d$ and $0 \leq l \leq L$. This allows us to trace the probability distribution on all positions (including boundaries) of the quantum walk.
		
	First, we can discuss the advantages of quantum walks over their classical counterparts, as discovered by Ambainis et al. in~\cite{ambainis2001one}. When given the initial state $\ket{s_1} \ket{R}$, the absorbing probability at position $0$ tends to $1/\sqrt{2}$ in the limit as $d \to \infty$, whereas in the classical case, the value is $1$. This property can be expressed as the MLTL formula $\vp_{0} = \lozenge \square \pair{M_{s_0}}{\ii_0}$, where $\ii_0 = [1/\sqrt{2}-\gamma, 1/\sqrt{2}+\gamma]$ and $\gamma>0$ is a given precision parameter.
  
	Next, we examine two properties of interest that demonstrate significant differences between quantum walks and their classical counterparts. To the best of our knowledge, these findings are new.
	\begin{enumerate}
		\item In the classical case, the absorbing probability at position $d$ is always smaller than 0.5 if the walk starts from the middle position and $d$ is even. However, this does not necessarily hold for quantum walks. Let $\vp_{1} = \square \pair{M_{s_d}}{\ii_{1}}$, where $\ii_{1} = [0,0.5)$, be the MLTL formula that expresses this property. Assuming the initial state is $\ket{s_{d/2}}\ket{R}$, we will see in Section~\ref{sec:experiments} that $\vp_{1}$ is false when $d=20$. 
			 
		\item Let $\ii_2=(0.4,1]$ and $\vp_{2} = \square(\pair{M_{s_{d-1}}}{\ii_{2}} \implies \pair{M_{s_1}}{\ii_{2}})$, which states that at any given time point, if the probability at position $d-1$ is larger than $0.4$, then the probability at position $1$ is also larger than $0.4$. In the classical case, $\vp_2$ is true due to the symmetry of the distribution over positions. However, as shown in Section~\ref{sec:experiments}, $\vp_2$ does not hold in the quantum case. Therefore, the distribution of the unbiased quantum walk over positions is asymmetric even when the walk starts from the middle position. 
	\end{enumerate}
\end{example}

\section{Model Checking Algorithm}
\label{sec:modelchecking}
In this section, we present an algorithm that can be used to approximately verify the satisfaction of $\g \mdst \vp$. For the convenience of the reader, we put all proofs of theoretical results in the appendix.

Using the notations in Eq.~\eqref{l1-fun}, we can formally define the model checking problem for $\trjst(\g)$ against MLTL formulas as follows.

\begin{problem}
\label{prob:quantumstate}
    Given a QMC $\g = (\h, \e, \rho_{0})$, a labeling function $\lfst$, and an MLTL formula $\vp$, the task is to decide whether $\g \mdst \vp$, which means determining whether $\lfst(\trjst(\g)) \in \l_\omega(\vp)$.
\end{problem}

Although MLTL extends LTL with quantum atomic propositions, the traditional model-checking techniques for LTL cannot be directly applied to solve Problem~\ref{prob:quantumstate}. This is because the state space of a QMC is continuous and uncountably infinite, even in the case where the state space is finite-dimensional. In contrast, classical LTL model checking deals with a discrete and finite state set. However, QMCs can simulate Markov chains as seen in Example~\ref{Exa:modeling_MC}, and interval LTL formulas in~\cite{agrawal2015approximate} can be represented by MLTL formulas as discussed in Section~\ref{Sec:logic}. Despite this, the counter-example presented in~\cite{agrawal2015approximate} shows that the language $\setnocond{\lfst(\trjst(\g))}$ is generally not $\omega$-regular. Therefore, the standard approach~\cite{vardi1999probabilistic} of model checking $\omega$-regular languages cannot directly solve Problem~\ref{prob:quantumstate} either.

To address this, we turn to the problem of approximate verification for QMCs, following the techniques introduced in~\cite{agrawal2015approximate}. The main idea in~\cite{agrawal2015approximate} involves studying the periodicity of states in \emph{finitely many} sub-chains obtained through BSCC decomposition~\cite{baier2008principles}, and a key property of Markov chains known as \emph{periodic stability}, which ensures their stability. However, extending this idea to the quantum setting is not straightforward. It has been proven that the BSCC decomposition of QMCs is not unique, but rather has \emph{infinitely many} possibilities~\cite{ying2013reachability,Baumgartner2011,guan2018decomposition} due to the continuous state space. Additionally, we will demonstrate below that QMCs do not exhibit periodic stability.

\begin{definition}
\label{def:e_{p}eriodicstable}
    A QMC $\g = (\h, \e, \rho_{0})$ is called \emph{periodically stable} if there exists an integer $\theta > 0$ such that 
        $\lim_{n \to \infty} \e^{n \theta}(\rho_{0}) = \rho^{*}$
    for some limiting quantum state $\rho^{*}$. 
    The smallest value of $\theta$, if it exists, is referred to as the \emph{period} of $\g$ and is denoted as $p(\g)$. Moreover, $\setcond{\e^{k}(\rho^*)}{0\leq k <p(\g)}$ are called the \emph{periodically stable states} of $\g$ as $\lim_{n \to \infty} \e^{n \theta}(\e^k(\rho_{0})) = \e^k(\rho^{*})$.
\end{definition}
In the classical case, any Markov chain $(\s, P, \mu_{0})$ is periodically stable~\cite{gallager2012discrete}. This means that there exists an integer $\theta > 0$ such that
        $\lim_{n \to \infty} \mu_{0} P^{n \theta} = \mu^{*}$
    for some limiting distribution $\mu^{*}$. 
    Furthermore, $\theta$ is independent of $\mu_{0}$.
    However, this property does not hold for QMCs, as demonstrated by the following example.
\begin{example}
\label{ex:u}
Let $\h = \sspan\setnocond{\ket{s_0}, \ket{s_1}}$ and $U = \ket{s_0} \bra{s_0} + e^{i 2 \pi \psi} \ket{s_1} \bra{s_1}$ be a unitary operator on $\h$, where $\psi$ is an irrational number. It can be proven that the QMC $(\h, \e_{U}, \rho_{0})$, where $\e_{U}(\rho) = U \rho U^{\dag}$, is not periodically stable for any generic initial state $\rho_0$.
In fact, a simple calculation reveals that 
\[
\e_{U}^{n \theta}(\rho_{0}) = \rho_{00} \cdot \ket{s_0} \bra{s_0} + \rho_{11} \cdot \ket{s_1} \bra{s_1} + e^{-i 2 \pi \psi n \theta} \rho_{01} \cdot \ket{s_0} \bra{s_1} + e^{i 2 \pi \psi n \theta} \rho_{10} \cdot \ket{s_1} \bra{s_0}
\]
where $\rho_{ij} = \bra{s_i} \rho \ket{s_j}$. It is worth noting that since $\psi$ is irrational, the set $\setcond{e^{i 2 \pi \psi m}}{m \in \naturals}$ is dense in the unit circle~\cite{hardy1979introduction}. Therefore, for any integer $\theta > 0$, the limit $\lim_{n \to \infty} \e_{U}^{n \theta}(\rho_{0})$ cannot exist, except when $\rho_{01} = \rho_{10} = 0$.
\end{example}

Note that the operator $\e_{U}$ in Example~\ref{ex:u} has four eigenvalues (with multiplicity taken into account): $1$, $1$, $e^{-i 2 \pi \psi}$, and $e^{i 2 \pi \psi}$. The corresponding eigenvectors are $\ket{s_0} \bra{s_0}$, $\ket{s_1} \bra{s_1}$, $\ket{s_0} \bra{s_1}$, and $\ket{s_1} \bra{s_0}$, respectively. We have proven that the system $(\h, \e_{U}, \rho_{0})$ is periodically stable if and only if the initial state $\rho_{0}$ does not have any components in the directions of $\ket{s_0} \bra{s_1}$ and $\ket{s_1} \bra{s_0}$. Interestingly, this is precisely why a QMC cannot be periodically stable (see Appendix~\ref{sec:periodicstable}). To put it in another way, a QMC is periodically stable only if the initial quantum state does not contain any components in the directions defined by the eigenvectors of the relevant super-operator (linear operator) corresponding to eigenvalues of the form $e^{i 2 \pi \psi}$ where $\psi$ is an irrational number. This result also offers an efficient method to verify the periodic stability of a given QMC $\g$ (and determine the period $p(\g)$). Specifically, symbolic computation can be used to check whether the eigenvalues of QMCs are irrational by analyzing their algebraic representations and mathematical properties. Such a method can be utilized to confirm the periodic stability of the quantum walk in Example~\ref{Exa:Quantum_walk}. Therefore, the approximate verification technique described in this paper can be applied to the quantum walk.

\begin{proposition}\label{Prop:periodicity}
    Given a QMC $\g=(\h,\e,\rho_0)$ with $\dim(\h)=d$, there is a way to check whether $\g$ is periodically stable with computational complexity $\bigO{d^{8}}$. If it is indeed stable, we can compute the period $p(\g)$ with a complexity of $\bigO{d^{8}}$.
\end{proposition}
The evaluation process for periodic stability, as described in Proposition~\ref{Prop:periodicity}, involves examining the eigenvalues and eigenvectors of the super-operator $\e$.
To calculate these eigenvalues and eigenvectors of $\e$ with Kraus operators $\setnocond{E_{k}}_{k}$, we can use the matrix representation $\m_{\e}$ ~\cite{ying2013verification} of $\e$, given by
$\m_{\e} = \sum_{k} E_{k} \otimes E_{k}^{*}.$
Here, $E_k^{*}$ represents the entry-wise complex conjugate of $E_k$. The linear operator (matrix) $\m_{\e}$ acts on $\h \otimes \h$. Based on this, we can derive the following lemma:

\begin{lemma}\label{Lem:eigenvalue}
    For a non-zero $A \in \lh$, $A$ is an eigenvector of $\e$ corresponding to the eigenvalue $\lambda$ if and only if $\ket{A}$ is an eigenvector of $\m_{\e}$ corresponding to the eigenvalue $\lambda$. In other words, $\e(A) = \lambda A$ if and only if $\m_{\e} \ket{A} = \lambda \ket{A}$.
\end{lemma}
In this context, $\ket{A} \in \h\otimes \h$ represents the vectorization of $A$, denoted by $\ket{A} \coloneqq (A \otimes I) \ket{\Omega}$. Here, $\ket{\Omega}$ denotes the (unnormalized) maximally entangled state on $\h \otimes \h$~\cite{nielsen2010quantum}. Assuming $\h = \sspan\setnocond{\ket{s_0}, \dotsc, \ket{s_{d-1}}}$, the maximally entangled state can be expressed as $\ket{\Omega} = \sum_{k=0}^{d-1} \ket{s_k} \ket{s_k}$. In particular, for a quantum state $\rho \in \dh$, we write $\ket{\rho} = (\rho \otimes I) \ket{\Omega}$. 

Now, our attention can be directed towards approximately model checking the QMC $\g$ that is periodically stable. To achieve this with a desired level of accuracy $\epsilon$, we can adopt the following approach. First and foremost, we need to identify the states of the chain that are periodically stable as defined in Definition~\ref{def:e_{p}eriodicstable}. After a sufficient number of steps, any state on the trajectory $\trjst(\g)$ will be close to one of the periodically stable states. This approach can also be applied to the labeled trajectory $L(\trjst(\g))$. By doing so, we can obtain an $\omega$-regular language. Following that, we can utilize the standard B\"uchi automata approach of model checking $\omega$-regular languages to analyze the obtained language. Therefore, in the following subsections, we will outline the step-by-step process for handling this.
\subsection{Periodically Stable States}\label{sec:Periodically_Stable_States}
Given a periodically stable QMC $(\h,\e,\rho_0)$, we can obtain the periodically stable states by employing a specific super-operator called the \emph{stabilizer} of $\g$, denoted by $\e_{\phi}$. This stabilizer is generated by the super-operator $\e$.

\begin{lemma}\label{Lem:stable_states}
    If a QMC $\g = (\h, \e, \rho_{0})$ is periodically stable with period $p(\g)$, then the set $\{\e_{\phi}(\e^{k}(\rho_{0}))\}_{0 \leq k < p(\g)}$ consists of the periodically stable states of $\g$. The computational complexity of obtaining such a set is $\bigO{d^{8}}$, where $d = \dim(\h)$.
\end{lemma}

The super-operator $\e_{\phi}$ is constructed from $\e$ by retaining the eigenvectors corresponding to eigenvalues with a magnitude of one, which do not vanish in the evolution $\e^n(\rho)$ as $n$ tends to infinity. Specifically, let $\m_{\e}$ be the matrix representation of $\e$ with Jordan decomposition $\m_{\e} = SJS^{-1}$, where $J = \bigoplus_{k=0}^{K} J_{k}(\lambda_{k}) = \bigoplus_{k=0}^{K} \lambda_{k} P_{k} + N_{k}$. Here, $\lambda_{k}$ represents the eigenvalues of $\m_{\e}$, $P_{k}$ is a projector onto the corresponding (generalized) eigenvector space, and $N_{k}$ is the corresponding nilpotent part. Furthermore, according to~\cite[Proposition 6.2]{wolf2012quantum}, the geometric multiplicity of any $\lambda_{k}$ with a magnitude of one (i.e., $|\lambda_{k}|=1$) is equal to its algebraic multiplicity, i.e., $N_{k} = 0$.
Then we define $J_{\phi} \coloneqq \bigoplus_{k : |\lambda_{k}| = 1} P_{k}$ as the projector onto the eigenspace corresponding to eigenvalues with magnitude one. By~\cite[Proposition 6.3]{wolf2012quantum}, it is confirmed that $\m_{\e_{\phi}}=SJ_{\phi} S^{-1}$ is indeed the matrix representation of some super-operator $\e_{\phi}$.

\subsection{Neighborhood of Quantum States}
Now we proceed to introduce the concepts of (symbolic) neighborhoods for (sequences of) quantum states using the labeling function $\lfst$.

\begin{definition}\label{def:neighborhood_state}
    Let $\rho \in \dh$ be a quantum state and $\epsilon > 0$. The (symbolic) $\epsilon$-neighborhood $\epsneigh(\rho)$ of $\rho$ is the subset of $2^{\ap}$ defined as
    \[
        \epsneigh(\rho) \coloneqq \setcond{\lfst(\rho')}{\rho' \in \dh, \norm{\rho - \rho'} < \epsilon},
    \]
    where $\norm{A} \coloneqq \sqrt{\tr(A^\dagger A)}$ represents the Schatten 2-norm (also known as the Hilbert–Schmidt norm) for the linear operator $A \in \lh$.
\end{definition}

Now we show that, after a certain number of steps, the symbols $L(\e^{n}(\rho_0))$ will be enclosed within the $\epsilon$-neighborhood of one of the periodically stable states.

\begin{lemma}
\label{mainlemmaQMC}
    Consider a periodically stable QMC $\g = (\h, \e, \rho_{0})$ with period $p(\g)$. Let $\eta_{k} = \e_{\phi}(\e^{k}(\rho_{0}))$, for each $0 \leq k < p(\g)$, as the periodically stable states of $\g$. 
    Then for any $\epsilon > 0$, there exists an integer $K^{\epsilon} > 0$ such that for any $n \geq K^{\epsilon}$,
    \[
        \lfst(\e^{n}(\rho_{0})) \in \epsneigh(\eta_{n \modulo p(\g)}).
    \]
    Moreover, the time complexity of computing $K^{\epsilon}$ is in $\bigO{d^{8}}$, where $d = \dim(\h)$.
\end{lemma}
 With Lemma~\ref{mainlemmaQMC}, we can define the concept of the (symbolic) neighborhood of trajectories for periodically stable QMCs.

\begin{definition}
\label{def:trajneigh}
    Given a periodically stable QMC $\g= (\h, \e, \rho_{0})$ and $\epsilon > 0$, the (symbolic) $\epsilon$-neighborhood of the trajectory $\trjst(\g)$ of $\g$ is defined to be the language $\epsneigh(\trjst(\g))$ over $(2^{\ap})^{\omega}$ such that $\xi \in \epsneigh(\trjst(\g))$ if and only if
    \begin{itemize}
    \item 
        $\xi[n] = \lfst(\e^{n}(\rho_{0}))$ for all $0 \leq n \leq K^{\epsilon}-1$ and
    \item 
        $\xi[n] \in \epsneigh(\eta_{n \modulo p(\g)})$ for all $n \geq K^{\epsilon}$,
    \end{itemize}
    where the states $\setcond{\eta_{k}}{0 \leq k <  p(\g)}$ and $K^{\epsilon}$ are as given in Lemma~\ref{mainlemmaQMC}.
\end{definition}

\subsection{Approximate Verification of Quantum Markov Chains}
Using Definition~\ref{def:trajneigh}, we can formulate and address the problem of approximate model checking for QMCs against MLTL formulas in the following manner.

\begin{problem}
\label{prob:approx}
    Given a periodically stable QMC $\g = (\h, \e, \rho_{0})$, a labeling function $\lfst$, an MLTL formula $\vp$, and $\epsilon > 0$, decide whether
    \begin{enumerate}
    \item 
         \emph{$\g$ $\epsilon$-approximately satisfies $\vp$ from below}, denoted $\g \models_{\epsilon} \vp$; that is, whether $\epsneigh(\trjst(\g)) \cap \l_{\omega}(\vp) \neq \emptyset$; 
    \item 
         \emph{$\g$ $\epsilon$-approximately satisfies $\vp$ from above}, denoted $\g \models^{\epsilon} \vp$; that is, whether $\epsneigh(\trjst(\g)) \subseteq \l_{\omega}(\vp)$. 
    \end{enumerate}
\end{problem}

To justify that Problem~\ref{prob:approx} is indeed an approximate version of Problem~\ref{prob:quantumstate}, we first note that $\lfst(\trjst(\g)) \in \epsneigh(\trjst(\g))$. 
Then we have three cases:
\begin{enumerate}
\item 
    if $\g \not\models_{\epsilon} \vp$, then $\epsneigh(\trjst(\g)) \cap \l_{\omega}(\vp)= \emptyset$, and hence $\lfst(\trjst(\g)) \notin \l_{\omega}(\vp)$ ($\g \not\models \vp$);
\item 
    if $\g \models^{\epsilon} \vp$, then $\epsneigh(\trjst(\g)) \subseteq \l_{\omega}(\vp)$, and hence $\lfst(\trjst(\g)) \in \l_{\omega}(\vp)$ ($\g \models \vp$);
\item 
    if neither $\g \not\models_{\epsilon} \vp$ nor $\g \models^{\epsilon} \vp$, then whether or not $\g \models \vp$ is unknown.
\end{enumerate}
If we find ourselves in the third scenario (the unknown case), we have the option to decrease the value of $\epsilon$ by half and then repeat the approximate model-checking process described in the previous scenarios. We can continue this process until we reach either of the first two cases, which will provide us with either a negative or affirmative answer to Problem~\ref{prob:quantumstate}. In exceptional cases, diminishing $\epsilon$ may not lead to termination. To prevent this, a predetermined number of iterations can be set for reducing epsilon. Determining when the procedure terminates seems difficult, and we would like to leave it as future work.

Finally, to solve Problem~\ref{prob:approx}, we represent $\epsneigh(\trjst(\g))$ in Definition~\ref{def:trajneigh} as the $\omega$-regular expression
\[
    \epsneigh(\trjst(\g)) = \setnocond{\lfst(\rho_{0})} \cdot \setnocond{\lfst(\e(\rho_{0}))} \cdots \setnocond{\lfst(\e^{K^{\epsilon}-1}(\rho_{0}))} \cdot \left(\epsneigh(\zeta_{0}) \cdots \epsneigh(\zeta_{p(\g)-1})\right)^{\omega}
\]
where $\zeta_{k} = \eta_{(K^{\epsilon} + k) \modulo p(\g)}$, $0 \leq k < p(\g)$,
and for any two sets $X$ and $Y$, $X \cdot Y = \setcond{xy}{x \in X, y \in Y}$.
Thus $\epsneigh(\trjst(\g))$ is $\omega$-regular and standard techniques~\cite{baier2008principles,HandbookMC18} can be employed to check $\epsneigh(\trjst(\g)) \cap \l_{\omega}(\vp) = \emptyset$ and $\epsneigh(\trj(\g)) \subseteq \l_{\omega}(\vp)$. 

\begin{algorithm}[t]
	\caption{ModelCheck($\g,  \ap, L, \vp, \epsilon$)}
	\label{ModelCheck}
	\begin{algorithmic}[1]
		\REQUIRE A periodically stable QMC $\g = (\h, \e, \rho_{0})$ with Kraus operators $\setnocond{E_{k}}_{k}$, a finite set of (measurement-based) atomic propositions $\ap$,  a labeling function $\lfst$,  an MLTL formula $\vp$, and $\epsilon> 0$.
		\ENSURE \TRUE{}, \FALSE{}, or \textbf{unknown}, where \TRUE{} indicates $\g \models \vp$, \FALSE{} indicates $\g \not\models \vp$, and \textbf{unknown} stands for an inconclusive answer.
		\STATE Put $\m_{\e}=\sum_{k} E_{k} \otimes E_{k}^{*}$ 
		\label{algo:line:matrixrepresentation}
		\STATE Get $p(\g)$ and $M_{\e_{\phi}} = S J_{\phi} S^{-1}$ by the Jordan decomposition form $\m_\e=SJS^{-1}$
		\STATE Put $\ket{\rho_{0}}=(\rho_{0}\otimes\ii)\ket{\Omega}$ 
		\FORALL{$k \in \setnocond{0, 1, \dotsc, p(\g) - 1}$}
		\STATE Set $\ket{\eta_{k}}$ to be $M_{\e_{\phi}}M^{k}_{\e}\ket{\rho_{0}}$
		\STATE Compute $\epsneigh(\eta_{k})$ by semi-definite programming
		\ENDFOR
		\label{algo:line:neighbourhood}
            \STATE{Get $K^\epsilon$ by solving some inequalities} \label{algo:line:trancation}
		\FORALL{$k \in \setnocond{0, 1, \dotsc, K^{\epsilon}-1}$}
		\label{algo:line:prefixbegin}
		\STATE Put $\rho_{k}=\e^k(\rho_{0})$ and compute $\lfst(\rho_{k})$
		\ENDFOR
		\label{algo:line:prefixend}
            \STATE Put $\zeta_{k} = \eta_{(K^{\epsilon} + k) \modulo p(\g)}$ for $0 \leq k < p(\g)$\label{algo:line:reorder}
		\STATE Let $\epsneigh(\trj(\g))$ be the $\omega$-regular language\label{algo:line:language}\\
		~\hfill$\setnocond{\lfst(\rho_{0})} \setnocond{\lfst(\rho_{1})} \cdots \setnocond{\lfst(\rho_{K^{\epsilon}-1})} \cdot (\epsneigh(\zeta_{0}) \epsneigh(\zeta_{1}) \cdots \epsneigh(\zeta_{p(\g)-1}))^{\omega}$\hfill~
		\STATE Construct the NBA $A_{\vp}$ for $\vp$ \COMMENT{standard construction}
		\label{algo:line:NBAformula}
		\STATE Construct the NBA $A_{\g}$ accepting $\epsneigh(\trj(\g))$\COMMENT{standard lasso-shaped construction}
		\label{algo:line:NBAlasso}
		\IF[standard B\"uchi automata operation]{$\l(A_{\g}) \cap \l(A_{\vp}) = \emptyset$}
		\label{algo:line:NBAintersectionEmptiness}
		\RETURN \FALSE
		\ELSIF[standard B\"uchi automata operation]{$\l(A_{\g}) \subseteq \l(A_{\vp})$}
		\label{algo:line:NBAinclusion}
		\RETURN \TRUE
		\ELSE
		\RETURN \textbf{unknown}
		\ENDIF
	\end{algorithmic}
\end{algorithm}

\begin{theorem} 
\label{thm:main}
The verification problem outlined in Problem~\ref{prob:approx} can be addressed using Algorithm~\ref{ModelCheck} within a time complexity of $\bigO{2^{\bigO{|\vp|}} \cdot (K^{\epsilon} + p(\g))+d^8}$.  Here, $d=\dim(\h)$, $|\vp|$ represents the size of MLTL formula $\vp$, and $p(\g)$ and $K^{\epsilon}$ are as given in Proposition~\ref{Prop:periodicity} and Lemma~\ref{mainlemmaQMC}, respectively.
\end{theorem}

Algorithm~\ref{ModelCheck} summarizes our techniques proposed above to answer Problem~\ref{prob:approx}. 
Starting from lines~\ref{algo:line:matrixrepresentation} to~\ref{algo:line:neighbourhood}, we make use of the Jordan decomposition of the matrix representation $\m_{\e}$ of $\e$ to determine the period $p(\g)$ and the stabilizer $\e_{\phi}$. These computations allow us to obtain the periodically stable states $\setnocond{\eta_{k}}_{k=0}^{p(\g)-1}$ of $\g$ and their corresponding symbolic $\epsilon$-neighborhood $\setnocond{\epsneigh{(\eta_{k})}}_{k=0}^{p(\g)-1}$ based on the given approximation level $\epsilon$. The steps involved in these computations utilize Proposition~\ref{Prop:periodicity} and Lemma~\ref{Lem:stable_states}. Afterwards, in line~\ref{algo:line:trancation}, we determine the truncation number $K^{\epsilon}$ using Lemma~\ref{mainlemmaQMC}. Subsequently, we compute $\setnocond{\lfst(\rho_{k})}_{k=0}^{K^{\epsilon}-1}$, which represents the symbols in $\ap$ of the first $K^\epsilon$ quantum states in the trajectory $\sigma(\g)$, from lines~\ref{algo:line:prefixbegin} to~\ref{algo:line:prefixend}. Finally, in lines~\ref{algo:line:reorder} and~\ref{algo:line:language}, we obtain an $\omega$-regular language $\epsneigh(\trj(\g))$ that represents the symbolic neighborhoods of the evolution $\trj(\g)$.
The B\"uchi automaton $A_{\vp}$ for the MLTL formula $\vp$ is constructed at line~\ref{algo:line:NBAformula} by means of a standard construction for an LTL formula $\vp$ (see, e.g.,~\cite{HandbookMC18}) while the B\"uchi automaton $A_{\g}$ at line~\ref{algo:line:NBAlasso} is obtained by an ordinary lasso-shaped construction: 
it is enough to insert a new state between each letter, make the state joining the stem and the lasso part accepting, and use the accepting state as the target of the last action in the lasso.
The two operations on B\"uchi automata at lines~\ref{algo:line:NBAintersectionEmptiness} and~\ref{algo:line:NBAinclusion} are standard:
intersection and emptiness reduce to automata product and strongly connected components decomposition, respectively, which require quadratic time (cf.~\cite{HandbookMC18}). 
Language inclusion, however, in general, requires exponential time and is PSPACE-complete (cf.~\cite{HandbookMC18});
in our case, however, we can remain in quadratic time by replacing the check $\l(A_{\g}) \subseteq \l(A_{\vp})$ with $\l(A_{\g}) \cap \l(A_{\neg\vp}) = \emptyset$, since constructing the B\"uchi automata $A_{\vp}$ and $A_{\neg\vp}$ requires the same asymptotic effort (cf.~\cite[Section~4.6.1]{HandbookMC18}).

\begin{remark}
The complexity of approximately verifying MCs against interval LTL formulas in the work of Agrawal et al.~\cite{agrawal2015approximate} is challenging to analyze and remains unknown. This is because the BSCC decomposition analysis, which is relied upon, is not an analytical method. Our model-checking algorithm, however, overcomes this by utilizing the Jordan normal form to create a linear-algebraic representation of the graph-theoretic (BSCC-based) model-checking algorithm mentioned in~\cite{agrawal2015approximate}. Consequently, our model-checking algorithm (Algorithm~\ref{ModelCheck}) can effectively verify MC models against interval LTL formulas. This capability stems from the ability of QMCs to emulate MCs, as illustrated in Example~\ref{Exa:modeling_MC}, and the greater expressiveness of our MLTL compared to interval LTL in~\cite{agrawal2015approximate}, as discussed in Section~\ref{Sec:logic}. Notably, MCs are inherently periodically stable, which extends to the modeled QMCs, making our model-checking algorithm suitable for this context. Additionally, while Agrawal et al.~\cite{agrawal2015approximate} did not offer a complexity analysis for their algorithm, our method establishes the first upper bound on the computational complexity of approximate model-checking for MCs against interval LTL specifications from their work. Specifically, when using our algorithm for model checking MCs, the complexity is reduced to  $\bigO{2^{\bigO{|\vp|}} \cdot (K^{\epsilon} + p(\g))+d^4}$, where $d$ represents the number of states in the MC.  This reduction is due to the fact that the complexity of computing the Jordan decomposition of the $d$-by-$d$ stochastic matrix $P$ in MCs is $\bigO{d^{4}}$. As a result, the time complexity of calculating $K^{\epsilon}$ and $p(\g)$ is also reduced to $\bigO{d^{4}}$.
\end{remark}

\section{Case Studies}
\label{sec:experiments}

To demonstrate the utility of the model-checking techniques proposed in this paper, we conducted case studies on quantum walks to investigate their temporal linear-time properties. We completed a prototype for implementing our model-checking algorithm and used it to automatically model check all MLTL formulas provided in Example~\ref{ex:qwc}. The prototype was built using Python for the quantum part to generate an $\omega$-regular language and Java for the classical part to model check the language against LTL formulas by calling Spot, a platform for LTL and $\omega$-automata manipulation~\cite{duret.22.cav}.

Before conducting the verification process, it is crucial to ensure the periodic stability of the QMC model for the quantum walk in Example~\ref{Exa:Quantum_walk}. By  Proposition~\ref{Prop:periodicity}, this can be achieved by computing the eigenvalues of the model and confirming that they only have 1 as the eigenvalue with a magnitude of one. Armed with this information, we can then employ Algorithm~\ref{ModelCheck} to verify the properties outlined in Example~\ref{ex:qwc}. The experimental results for these verifications can be found in Table~\ref{tb:experiments}. 

The first experiment in Table~\ref{tb:experiments} confirms the advantages of quantum walks over classical random walks, which was previously established by Ambainis et al. in~\cite{ambainis2001one}. The remaining two experiments aim to uncover new properties of quantum walks. Moreover, we also utilized Algorithm~\ref{ModelCheck} to verify these same properties for one-dimensional unbiased classical random walks with absorbing boundaries, which yielded contrasting results (\textbf{false} for the first experiment and \textbf{true} for the last two experiments in Table~\ref{tb:experiments}). This confirms that these properties are unique phenomena specific to quantum walks. 

 \begin{table}[tpb]
 	\caption{The experimental result for position number $d=20$ and different initial states $\rho_0=\ketbra{\psi_0}{\psi_0}$.}
 	\label{tb:experiments}
 	\centering
 	\renewcommand{\arraystretch}{1.2}
 	\setlength{\tabcolsep}{5pt}
 	\begin{tabular}{c|c|ccc}
 		\multirow{2}{11mm}{\centering  $\ket{\phi_{0}}$} & \multirow{2}{*}{Formula $\vp$} & \multicolumn{3}{c}{Parameter $\epsilon$} \\
 		& & 0.5 & 0.25 & 0.125 \\ 
 		\hline
       $\ket{s_{1}}\ket{R}$ & $\lozenge \square \pair{M_{s_0}}{[1/\sqrt{2}-0.1, 1/\sqrt{2}+0.1]}$ & \textbf{unknown} & \textbf{unknown} & \textbf{true} \\
       \hline
 		\multirow{2}{*}{$\ket{s_{10}}\ket{R}$} & $\square \pair{M_{s_{20}}}{[0,0.5)}$ & \textbf{unknown} & \textbf{unknown} & \textbf{false}\\
 		 & $\square(\pair{M_{s_{19}}}{(0.4,1]} \implies \pair{M_{s_1}}{(0.4,1]})$ & \textbf{unknown} & \textbf{unknown} & \textbf{false}\\
 	\end{tabular}
 \end{table}

\section{Conclusion}  
In this paper, we proposed a new quantum logic called measurement-based linear-time temporal logic (MLTL) to specify the quantitative properties of quantum Markov chains (QMCs). For model checking MLTL formulas, an algorithm was developed. The measurement-based atomic propositions of MLTL build upon the subspace-based quantum atomic propositions introduced by Birkhoff and von Neumann~\cite{birkhoff1936logic}. Furthermore, we demonstrated the practical applicability of our model-checking algorithm in quantum walks. This not only confirms the previously established advantages discovered by Ambainis et al.~\cite{ambainis2001one} of quantum walks over classical random walks, but also uncovers new phenomena that are unique to quantum walks. 

As future work, we note that the quantum walk in Example~\ref{Exa:Quantum_walk} can also be written as a while-loop quantum program~\cite{ying2013verification}, and the absorbing probabilities are exactly the termination probabilities of the program. Indeed, any quantum program written in the \textsc{While} language can be modeled by a QMC~\cite{ying2016foundations}. So our model-checking approach can be naturally applied in the verification of program properties specified as MLTL formulas. Therefore, we plan to apply our model-checking algorithms to verify quantum programs. Moreover, it would be intriguing to broaden the application of the methods presented in this paper to verify the behavior of non-periodically stable QMCs.  One possible approach is that any non-periodically stable QMC will be close to a periodically stable one with arbitrary precision, as any irrational number (eigenvalue) will be close to a rational number (eigenvalue) with the same precision.

\bibliographystyle{unsrt}
\bibliography{note110418}
\clearpage
\section*{Appendix}

\appendix

\section{Visualization of the Evolution of Quantum Walks}\label{appendix:quantum_walk}
\begin{figure}
    \centering
    \resizebox{\linewidth}{!}{
    \includegraphics{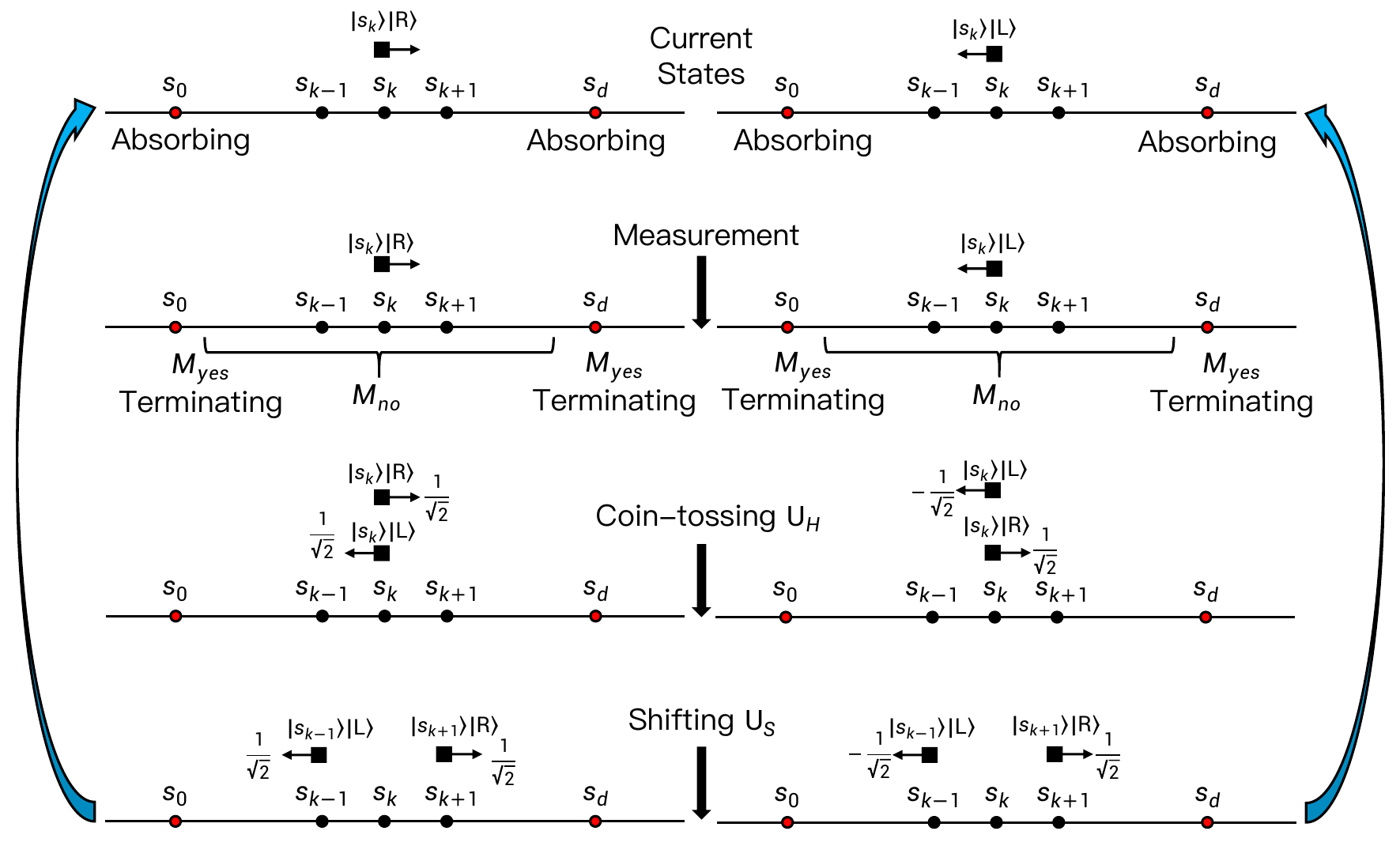}
    }
    \caption{The dynamics of one-dimensional quantum walk with boundaries in one step}
    \label{fig:quantum_walks}
\end{figure}
To gain a better understanding of the evolution of quantum walks, it is necessary to introduce the concept of \emph{amplitude distribution} from quantum physics. As stated in Section~\ref{Sec:preliminaries}, in a Hilbert space $\h = \sspan\setnocond{\ket{s_{0}}, \dotsc, \ket{s_{d}}}$, any quantum pure state $\ket{\psi}$ can be expressed as $\ket{\psi} = \sum_{k=0}^{d} a_{k}\ket{s_{k}}$, where $a_{k}$ is a complex number and referred to as the \emph{amplitude} of $\ket{\psi}$. 
The amplitude distribution $(a_{0}, \dotsc, a_{d})$ is a unit complex vector ($a_{0} a_{0}^{*} + \ldots + a_{d} a_{d}^{*} = 1$) formed by the set of amplitudes $\setnocond{a_{0}, \dotsc, a_{d}}$. It is important to note that the amplitude distribution can also be used to represent a probability distribution ($a_{0} a_{0}^{*}, \dotsc, a_{d} a_{d}^{*}$). Conversely, a probability distribution $(p_{0}, \dotsc, p_{d})$ can be expressed as an amplitude distribution $(\sqrt{p_{0}}, \dotsc, \sqrt{p_{d}})$.

To further illustrate the dynamic nature of the quantum walk, we can examine the evolution of the state at each step of the walk. This can be visualized in Fig.~\ref{fig:quantum_walks}, where the state at position $s_k$ is shown. The cases at other positions follow a similar pattern. Initially, the current state of the quantum system is represented as $\ket{s_k}\ket{R}$ (on the left side of Fig.~\ref{fig:quantum_walks}) or $\ket{s_{k}}\ket{L}$ (on the right side of Fig.~\ref{fig:quantum_walks}), indicating the state at position $s_{k}$ with the direction either right or left, respectively.

To begin the walk, we perform a measurement $\setnocond{M_{yes}, M_{no}}$ on the current state to determine if the position $s_{k}$ corresponds to the absorbing positions $s_0$ or $s_{d}$. If it is an absorbing position, the walk terminates. If not, the walk proceeds to the next two sub-steps, which involve moving the current position.

Next, a coin (Hadamard operator) $U_{H}$ is tossed to decide the direction of movement. This generates two possibilities with different directions and amplitude distributions, determined by the current state. Specifically, if the current state is $\ket{s_k}\ket{R}$ (on the left side of Fig.~\ref{fig:quantum_walks}), the amplitude distribution is $(\frac{1}{\sqrt{2}},\frac{1}{\sqrt{2}})$, representing the quantum walk moving to the right and left with a probability distribution of $(\frac{1}{2},\frac{1}{2})$. Similarly, for the state $\ket{s_{k}}\ket{L}$ (on the right side of Fig.~\ref{fig:quantum_walks}), the amplitude distribution is $(\frac{1}{\sqrt{2}},-\frac{1}{\sqrt{2}})$, indicating the quantum walk moving to the right and left with a probability distribution of $(\frac{1}{2},\frac{1}{2})$.

After determining the direction, the walk moves one position to the right or left based on the amplitude distribution by performing a shift operator $U_S$. This completes one step of the walk, and the process continues with the updated current state.

It is important to note that the elements of the amplitude distribution can be negative, such as $-\frac{1}{\sqrt{2}}$. As a result, when computing the total amplitude of the walk, subtractions occur, leading to some amplitudes decreasing in absolute value (probability) while others increase. This phenomenon, known as \emph{quantum interference}, results in the quantum walk exhibiting behavior that is significantly different from classical random walks. For example, the unbiased walk in Example~\ref{Exa:Quantum_walk} starting at the central position will exhibit a non-symmetric evolution. Our model-checking algorithm has identified and verified several properties related to this behavior.

\section{Proof of Proposition~\ref{Prop:periodicity}}
\label{sec:periodicstable}

In this section, we give an easily checkable characterization of the periodic stability of QMCs to complete the proof of Proposition~\ref{Prop:periodicity}.

Before proceeding, let us establish some mathematical notations from linear algebra that will be used later.

Firstly, it should be noted that any linear map $T$ on $\lh$ can have a maximum of $d^{2}$ (where $d = \dim(\h)$) complex eigenvalues $\lambda$ that satisfy $T(A) = \lambda A$ for some non-zero operator $A \in \lh$. 
We denote the \emph{spectrum} of $T$ as $\spec(T)$, which represents the set of all eigenvalues of $T$. 
The spectral radius of $T$ is defined as $\radius{T} = \max\setcond{|\lambda|}{\lambda \in \spec(T)}$. 
In particular, for any super-operator $\e$ on $\h$, it holds that $\radius{\e} = 1$~\cite[Proposition~6.1]{wolf2012quantum}. 
We denote the set of eigenvalues of $T$ with a magnitude of one as $\maxmag{T} = \setcond{\lambda \in \spec(T)}{|\lambda| = 1}.$
It is important to note that the calculation of $\spec(\e)$ can be simplified to the calculation of $\spec(\m_{\e})$, as shown in Lemma~\ref{Lem:eigenvalue}.  Furthermore,
the \emph{support}  of a quantum state $\rho$, denoted as $\supp{\rho}$, is a subspace of $\h$ spanned by the eigenvectors corresponding to the non-zero eigenvalues of $\rho$. If $\supp{\rho} = \h$, then $\rho$ is referred to as a \emph{full-rank} state.  $\rho$ is considered a \emph{stationary state} of $\e$ if and only if $\e(\rho) = \rho$. Additionally, a stationary state $\rho$ is classified as  \emph{maximal} if for any other stationary state $\sigma$ of $\e$, the support of $\sigma$ is a subspace of the support of $\rho$, i.e., $\supp{\sigma}\subseteq \supp{\rho}$. It is important to note that, for a given super-operator $\e$, the support of any maximal stationary state of $\e$ remains the same~\cite{guan2018decomposition}.

Let $(\h, \e, \rho_{0})$ be a QMC where $\m_{\e}$ represents the matrix $\e$ and $\m_{\e} = SJS^{-1}$ is its Jordan decomposition.
For each $0 \leq k \leq d^{2}-1$ where $d=\dim(\h)$, let $\ket{s_{k}}$ be the $k$-th column vector of $S$.
Since $S$ is invertible, the vectors $\ket{s_{k}}$ form a basis of the Hilbert space $\h \otimes \h$. Therefore, for any quantum state $\rho \in \dh$, its vectorization $\ket{\rho}\in\h\otimes\h$ can be uniquely represented as a linear combination of these vectors: 
$ \ket{\rho} = \sum_{k} a_{k} \ket{s_{k}}$. 
Let
\[
	\maxmag{\rho} = \setcond{\lambda \in \maxmag{\e}}{\text{$\m_{\e} \ket{s_{k}} = \lambda \ket{s_{k}}$ for some $k$ with $a_{k} \neq 0$}}
\]
be the set of eigenvalues of $\m_{\e}$ with a magnitude of one that contribute non-trivially to $\ket{\rho}$.  

We can now state the following proposition, which provides a method to determine whether a QMC $\g$ is periodically stable, as mentioned in Proposition~\ref{Prop:periodicity}. 
\begin{proposition}
\label{prop:periodicstable}
A QMC $(\h, \e, \rho_{0})$ is periodically stable if and only if $\maxmag{\rho_{0}}$ does not contain any element of the form $e^{i 2 \pi \psi}$, where $\psi$ is an irrational number. Furthermore, the complexity of checking periodic stability is $\bigO{d^8}$, where $d = \dim(\h)$.
\end{proposition}
\begin{proof}
  First, note that for any $m > 0$ and $\rho \in \dh$, $\ket{\e^{m}(\rho)} = M^{m}_{\e} \ket{\rho}$. 
  Thus $(\h, \e, \rho_{0})$ is periodically stable if and only if there exists an integer $\theta > 0$ such that $\lim_{n \to \infty} \m_{\e}^{n \theta} \ket{\rho_{0}}$ exists. 
  Let $\ket{\rho_{0}} = \sum_{j} a_{j} \ket{s_{j}}$. 
  For any $0 \leq j \leq d^{2}-1$, if $\ket{s_j}$ is an (generalized) eigenvector of $\m_{\e}$ corresponding to an eigenvalue with magnitude strictly smaller than $1$, then $\lim_{n \to \infty} \m_{\e}^{n\theta} \ket{s_{j}} = 0$ for any $\theta$. 
  Thus we only need to care about $\ket{s_j}$ corresponding to eigenvalues with a magnitude of one. Following~\cite[Lemma 2]{guan2017super}, $\p \circ \e$ shares with $\e$ the same eigenvalues with a magnitude of one and the corresponding eigenvectors, where $\p(\rho) = P \rho P$ for all $\rho$, and $P$ is the projector onto the support of the maximal stationary state.
  Therefore, w.l.o.g., we assume that $\e$ has a full-rank stationary state. 
  This kind of $\e$ is called \emph{faithful} in~\cite{Albert2018}.

  Furthermore, for faithful $\e$, the Kraus operators $\setnocond{E_{i}}$ admit a diagonal form with respect to an appropriate decomposition $\h = \oplus_{k=0}^{t-1} \h_{k}$~\cite{ying2013reachability,Burgarth2013}, where $\oplus$ denotes the \emph{direct sum} operation.
  To be specific, for all $i$,
  \[
    E_{i} = \oplus_{k=0}^{t-1} E_{i,k} = 
      \left[
        \begin{array}{cccc}
          E_{i,0} & & &\\
          & E_{i,1} & & \\
          & & \ddots & \\
          & & & E_{i,t-1} 
        \end{array}
      \right]
  \]
  where $E_{i,k}\in \l(\h_{k})$,
  so $\m_{\e}$ has the corresponding structure
  \[
    \m_{\e} = \bigoplus_{k,l} \m_{k,l}
  \]
  where $\m_{k,l} = \sum_{i} E_{i,k} \otimes E_{i,l}^{*}$. Furthermore, for any $\theta >0$ we have
  \begin{equation}
    \lim_{n \to \infty} \m_{\e}^{n\theta} \ket{\rho_{0}} = \bigoplus_{k,l} \lim_{n \to \infty} \m_{k,l}^{n\theta} \ket{\rho_{0,k,l}},
    \label{eq:tmp}
  \end{equation}
  where $\ket{\rho_{0,k,l}}$ is the restriction of $\ket{\rho_{0}}$ onto $\h_{k} \otimes \h_{l}$, i.e., $\ket{\rho_{0}} = \oplus_{k,l} \ket{\rho_{0,k,l}}$.
  Now it is easy to see that $\lim_{n \to \infty} \m_{\e}^{n\theta} \ket{\rho_{0}}$ exists if and only if for any $k$ and $l$, $\lim_{n \to \infty} \m_{k,l}^{n\theta} \ket{\rho_{0,k,l}}$ exists.

  To verify the existence of $\lim_{n \to \infty} \m_{k,l}^{n\theta} \ket{\rho_{0,k,l}}$, let 
  \[
    \ket{\rho_{0}} = \sum_{j} a_{j} \ket{s_{j}} = \bigoplus_{k,l} \sum_{j} a_{j} \ket{s_{j,k,l}},
  \]
  where $\ket{s_{j,k,l}}$ is the restriction of $\ket{s_{k}} $ onto $\h_{k} \otimes \h_{l}$. 
  Then $\ket{\rho_{0,k,l}} = \sum_{j} a_{j} \ket{s_{j,k,l}}$. Define 
  \[
    \maxmag{\rho_{0,k,l}} = \setcond{\lambda \in \maxmag{\m_{k,l}}}{\text{$\m_{k,l} \ket{s_{j,k,l}} = \lambda \ket{s_{j,k,l}}$ for some $j$ with $a_{j} \neq 0$}}.
  \]
  For any $0 \leq j \leq d^{2}-1$ and $0 \leq k,l \leq m-1$, we have two cases to consider:
  \begin{itemize}
  \item 
    if $\ket{s_{j,k,l}}$ is a (generalized) eigenvector of $\m_{k,l}$ corresponding to an eigenvalue with magnitude strictly smaller than $1$, then $\lim_{n \to \infty} \m_{\e}^{n\theta_{j}} \ket{s_{j,k,l}} = 0$ for any $\theta_{j}$; 
  \item 
    if $\ket{s_{j,k,l}}$ is an eigenvector of $\m_{k,l}$ corresponding to an eigenvalue with a magnitude of one, then $\m_{k,l}^{n} \ket{s_{j,k,l}} = e^{i 2 \pi \psi_{j,k,l} n} \ket{s_{j,k,l}}$ for some $\psi_{j,k,l}$. 
    Thus we have $\lim_{n \to \infty} \m_{\e}^{n\theta_{j}}\ket{s_{j,k,l}}$ exists for some $\theta_{j} > 0$ if and only if $\psi_{j,k,l}$ is rational and $\theta_{j} \psi_{j,k,l}$ is an integer~\cite{hardy1979introduction}.
  \end{itemize}

  Note that the matrices $\m_{k,l}$'s have the following spectral properties (cf.~\cite{guan2018structure}):
  for any $k$ and $l$, $\maxmag{\m_{k,l}} = \emptyset$ or $\maxmag{\m_{k,l}} = \setnocond{\exp(i 2 \pi (r + \psi_{k,l})/N_{k,l})}_{r=0}^{N_{k,l}-1}$, where $N_{k,l}$ is a positive integer and $\psi$ is a real number. 
  Thus $\lim_{n \to \infty} \m_{k,l}^{n\theta} \ket{\rho_{0,k,l}}$ exists if and only if $\maxmag{\rho_{0,k,l}}$ does not contain any element of the form $e^{i 2 \pi \psi}$ for some irrational number $\psi$. 
  We complete the proof by noting that $\maxmag{\rho} = \setcond{\lambda \in \maxmag{\rho_{0,k,l}}}{0 \leq k,l \leq m-1}$.
  \qed
\end{proof}
%


It is important to highlight that the proof of Proposition~\ref{prop:periodicstable} offers a method to check the periodic stability of QMCs. The main computational cost lies in the Jordan decomposition $\m_{\e}=SJS^{-1}$ of $\m_{\e}$. It is widely recognized that the complexity of the Jordan decomposition for an $n$-by-$n$ matrix is $\bigO{n^4}$. Therefore, for a $d^2$-by-$d^2$ matrix $\m_{\e}$, the cost is $\bigO{d^8}$. Consequently, the aforementioned proof of Proposition~\ref{prop:periodicstable} concludes the proof of the first part of Proposition~\ref{Prop:periodicity}.

\begin{corollary}
\label{cor:pstable}
	Let $\g = (\h, \e, \rho_{0})$ be a periodically stable QMC with $d = \dim(\h)$. 
	\begin{itemize}
	\item 
		If we express $\maxmag{\rho_{0}}$ as
		\[
			\maxmag{\rho_{0}} = \setcond{e^{2 \pi i p_{k}/q_{k}}}{\text{$p_{k}$ and $q_{k}$ are coprime positive integers}}_{k},
		\]
		then $p(\g) = \lcm\setnocond{q_{k}}_{k}$,  the least common multiple of $\setnocond{q_{k}}_{k}$.
	\item	
		For any positive integer $\theta$, $\lim_{n \to \infty} \e^{n\theta}(\rho_{0})$ exists if and only if $p(\g)$ divides $\theta$.
	\end{itemize}
\end{corollary}
Corollary~\ref{cor:pstable} presents a technique for calculating $p(\g)$ by determining the Jordan decomposition of a $d^2$-by-$d^2$ matrix $\m_{\e}$ and then evaluating $\maxmag{\rho_{0}}$ using its definition. The time complexity of this method is $\bigO{d^{8}}$. This concludes the proof of the second part of Proposition~\ref{Prop:periodicity}.


\section{Proof of Lemma~\ref{Lem:eigenvalue}}
\begin{proof}[of Lemma~\ref{Lem:eigenvalue}]
To begin, let us establish the transformation of $\e$ on $\h_1$ and its corresponding matrix representation on $\h_1\otimes \h_2$ as shown below:
\begin{equation}\label{Eq:partial_trace}
   \text{$\e(A) = \tr_{2}(\m_{\e}(A \otimes I) \ket{\Omega} \bra{\Omega})$, for any $A \in \mathcal{L}(\h_1)$.}
\end{equation}
Here, $\tr_{2}$ denotes the \emph{partial trace} operation on the second Hilbert space~\cite{nielsen2010quantum}. Mathematically, $\tr_{2}(\functiondot)$ is a linear map from $\l(\h_1 \otimes \h_2)$ to $\l(\h_1)$. Given that $\h = \sspan\setnocond{\ket{s_0}, \dotsc, \ket{s_{d-1}}}$, we can express $\tr_{2}(\functiondot)$ as follows:
\[\text{$\tr_{2}(\m) = \sum_{k} (I \otimes \bra{s_k}) \m (I \otimes \ket{s_k})$, for all $\m \in \l(\h_1 \otimes \h_2)$.}\]
Here, $I$ represents the identity operator on $\h_1$. By utilizing the transformation in Eq.~\eqref{Eq:partial_trace}, we can readily prove that $\e(A) = \lambda A$ if and only if $\m_{\e} \ket{A} = \lambda \ket{A}$, where $A \in \lh$ is not a zero matrix.
\qed
\end{proof}

\section{Proof of Lemma~\ref{Lem:stable_states}}
\begin{proof}[of Lemma~\ref{Lem:stable_states}]
  Let $\theta = p(\g)$.
  Based on the proof of Proposition~\ref{prop:periodicstable}, Corollary~\ref{cor:pstable}, and the construction of the stabilizer $\e_{\phi}$, we can conclude that \[\text{$\lim_{n \to \infty} \e^{n\theta+k}(\rho_{0})=\e_{\phi}(\e^{k}(\rho_{0}))$ for any integer $0 \leq k \leq \theta-1$}.\]	
 In other words, the set $\{\e_{\phi}(\e^{k}(\rho_{0}))\}_{0 \leq k < p(\g)}$ consists of the periodically stable states of $\g$. The process of obtaining this set has the same computational complexity as that of constructing $\e_{\phi}$, which is $\bigO{d^{8}}$  for the Jordan decomposition of the $d^2$-by-$d^2$ matrix $\m_{\e}$, as discussed in Subsection~\ref{sec:Periodically_Stable_States}.
 \qed
\end{proof}

\section{Proof of Lemma~\ref{mainlemmaQMC}}
To prove Lemma~\ref{mainlemmaQMC}, we need to characterize the $\epsilon$-neighborhoods of the quantum state $\rho \in \dh \subseteq \lh$ and the super-operator $\e \in \bbh$. Here, $\bbh$ represents the set of all linear operators from $\lh$ to $\lh$. To do this, we must introduce norms on $\lh$ and $\bbh$. The Schatten 2-norm (also known as the Hilbert-Schmidt norm) on $\lh$ has already been introduced in Definition~\ref{def:neighborhood_state} to define the $\epsilon$-neighborhood of $\rho$. Now, we introduce an additional norm on $\bbh$, and we choose to utilize the Schatten 2-norm on $\lh$ to induce an operator norm on $\bbh$. It is important to note that the results presented here are applicable to any other norm because all norms in a finite-dimensional Hilbert space are equivalent~\cite{Horn2013}.

\begin{definition}
    Given a Hilbert space $\h$, the operator norm $\opnorm{\functiondot}$ on $\bbh$ induced by the Schatten 2-norm $\norm{\functiondot}$ is defined as follows:
    \[
        \text{$\opnorm{T} \coloneqq \sup \setcond{\norm{T(A)}}{\text{$A \in \lh$ with $\norm{A} = 1$}}$ for any $T \in \bbh$.}
    \]
\end{definition}

Furthermore, for simplicity, we will refer to $\opnorm{\functiondot}$ as $\norm{\functiondot}$ unless there is potential confusion.

Based on the given norm and the fact that $\norm{A} = \sqrt{\tr(A^\dagger A)} = \sqrt{\braket{A}{A}}$, it can be easily shown that for any super-operator $\e$, 
\[
    \norm{\e} = \norm{\m_{\e}} = \max_{\lambda \in \spec(\m_{\e}^{\dag} \m_{\e})} \sqrt{\lambda}. 
\]
In other words, $\norm{\e}$ represents the maximum singular value of $\m_{\e}$, which is equal to 1 for the super-operator $\e$ (i.e., $\norm{\e}=1$). 
With this, for any super-operators $\e_{1}$, $\e_{2}$, and $\f$, we have
\begin{equation}
    \norm{(\e_{1} - \e_{2}) \circ \f} \leq \norm{\e_{1}-\e_{2}} \cdot \norm{\f} = \norm{\e_{1} - \e_{2}}.
    \label{eq:contractive}
\end{equation}
The inequality is a result of the sub-multiplicative property of the operator norm, which states that $\norm{T_{1} T_{2}} \leq \norm{T_1} \cdot \norm{T_2}$.

Furthermore, for any $\rho \in \dh$ and super-operators $\e_{1}$ and $\e_{2}$, we have
\begin{equation}
    \norm{\e_{1}(\rho) - \e_{2}(\rho)} \leq \norm{\e_{1} - \e_{2}} \cdot \norm{\rho} \leq \norm{\e_{1} - \e_{2}}.
    \label{eq:operatornorm}
\end{equation}
The second inequality follows from the fact that $\norm{\rho} \leq 1$.

To prove Lemma~\ref{mainlemmaQMC}, an essential result of the asymptotic property of $\e$ characterized by $\e_{\phi}$ from~\cite{wolf2012quantum} is also needed. This result is based on the \emph{Jordan condition number}. 

It is important to note that the Jordan decomposition $\m_{\e}=SJS^{-1}$ is not unique, and we define 
\[
\alpha(\e) = \inf_{S} \setcond{\norm{S} \cdot \norm{S^{-1}}}{\text{$S^{-1} \m_\e S$ is in Jordan normal form}}
\]
as the \emph{Jordan condition number}~\cite{wolf2012quantum} of $\e$.

\begin{lemma}[{cf.~\cite[Theorem 8.23]{wolf2012quantum}}]
	\label{lem:inequality}
	For any $n > 0$ we have
	\[
	C^{-1}\omega^{n} n^{d_{\omega} - 1} \leq \norm{\m_{\e}^{n} - \m_{\e_{\psi}}^{n}} \leq C\omega^{n} n^{d_{\omega}-1}
	\]
	where $\e_{\psi} = \e \circ \e_{\phi}$, $\omega = \sup\setcond{|\lambda|}{\lambda \in \spec(\e), |\lambda| <1}$ is the largest modulus of eigenvalues of $\e$ in the interior of the unit disc, $d_{\omega}$ is the dimension  of the largest Jordan block corresponding to eigenvectors of modulus $\omega$, and
	\[
	C = 
	\begin{cases}
	\alpha(\e) & \text{if $d_{\omega} = 1$,}\\
	\alpha(\e)(\omega(d_{\omega} - 1))^{ d_{\omega}-1} & \text{if $1 < d_{\omega} \leq n+1.$} 
	\end{cases}
	\]		
\end{lemma} 
Now, we present the proof of Lemma~\ref{mainlemmaQMC}.

\begin{proof}[of Lemma~\ref{mainlemmaQMC}]
  Let $\theta = p(\g)$.
  By Lemma~\ref{Lem:stable_states},  $\{\e_{\phi}(\e^{k}(\rho_{0}))\}_{0 \leq k < p(\g)}$ are the periodically stable states of $\g$, i.e.,
  \[
	\lim_{n \to \infty} \e^{n}(\rho_{0})=\e_{\phi}(\e^{n \modulo \theta }(\rho_{0})).
  \]
  Thus for any $\epsilon>0$, there exists a positive integer $K^{\epsilon}$ such that for all $n>K^{\epsilon}$, 
  \[\norm{ \e^{n}(\rho_{0})-\e_{\phi}(\e^{n \modulo \theta }(\rho_{0}))}<\epsilon\]
  as desired.

  To determine $K^{\epsilon}$, we recall from Lemma~\ref{lem:inequality} that
  \[
    C^{-1} \omega^{n} n^{d_{\omega}-1} \leq \norm{\m_{\e}^{n} - \m_{\e_{\psi}}^{n}} \leq C \omega^{n} n^{d_{\omega}-1}.
  \]
  Let $n=m\theta+k$ with $0\leq k\leq\theta-1$, and  note that $\e_{\psi}^{m \theta}(\rho_{0}) = \e_{\phi}(\rho_{0})$ by Corollary~\ref{cor:pstable}. 
  We have
  \[\norm{\e^{m \theta+k}(\rho_{0}) - \e_{\phi}(\e^k(\rho_{0}))} = \norm{\e^{m \theta+k}(\rho_{0}) - \e_{\psi}^{m \theta+k}(\rho_{0})} \leq \norm{\e^{m\theta} - \e_{\psi}^{m\theta}}\]
  where the inequality follows from Eqs.~\eqref{eq:contractive} and~\eqref{eq:operatornorm}. 
  So we can simply set $K^{\epsilon}$ to be the minimal integer satisfying 
  \begin{eqnarray}
    \text{$C \omega^{K^\epsilon} {(K^\epsilon)}^{d_{\omega}-1} < \epsilon$ \quad and \quad $K^\epsilon + 1 > d_{\omega}$},
    \label{Eq_Me}
  \end{eqnarray}
  where the second inequality comes from the requirement of $C$ in Lemma~\ref{lem:inequality}. 
  Finally, the computation of $C$ boils down to the Jordan decomposition of $\e$, which makes the time complexity of calculating $K^{\epsilon}$ to be $\bigO{d^{8}}$.
  \qed
\end{proof}

\section{Proof of Theorem~\ref{thm:main}}
\begin{proof}[of Theorem~\ref{thm:main}]
To see the correctness of Algorithm~\ref{ModelCheck},  we can proceed with a detailed examination of its steps immediately following the presentation of Theorem~\ref{thm:main}. 

To analyze the complexity of Algorithm~\ref{ModelCheck}, we first focus on the computation of the $\omega$-regular language $\epsneigh(\trj(\g))$ from lines~\ref{algo:line:matrixrepresentation} to~\ref{algo:line:language}. The main computational tasks include determining the period $\p(\g)$, the stabilizer $\e_{\phi}$, the periodically stable state $\setnocond{\eta_k}_{k=0}^{p(\g)-1}$, and the truncation number $K^\epsilon$. These values can be obtained using Proposition~\ref{Prop:periodicity} and Lemmas~\ref{Lem:stable_states} and~\ref{mainlemmaQMC}. The overall complexity is $\bigO{d^8}$, with the main cost being the Jordan decomposition of a $d^2$-by-$d^2$ matrix $\m_{\e}$. It is well-known that the complexity of the Jordan decomposition for an $n$-by-$n$ matrix is $\bigO{n^4}$. Therefore, for a $d^2$-by-$d^2$ matrix $\m_{\e}$, the cost is $\bigO{d^8}$.

Once we have obtained the $\omega$-regular language $\epsneigh(\trj(\g))$, we can utilize the standard B\"uchi automata approach from Line~\ref{algo:line:NBAformula} to the end of Algorithm~\ref{ModelCheck} to perform model checking on $\epsneigh(\trj(\g))$ against a LTL (MLTL) formula $\vp$. This approach allows us to address the verification problem outlined in Problem~\ref{prob:approx}. 
As said in the analysis after the statement of Theorem~\ref{thm:main}, the main source of complexity is the model checking of $\epsneigh(\trj(\g))$ against the formulas $\vp$ and $\neg\vp$: the former is used at line~\ref{algo:line:NBAintersectionEmptiness} while the latter at line~\ref{algo:line:NBAinclusion} since checking $\l(A_{\g}) \subseteq \l(A_{\vp})$ is equivalent to verifying $\l(A_{\g}) \cap \l(A_{\neg\vp}) = \emptyset$. 
These operations are done by constructing the B\"uchi automata $A_{\vp}$ and $A_{\neg\vp}$ for $\vp$ and $\neg\vp$, respectively, and checking the emptiness of their language intersection with $\l(A_{\g})$, i.e., with $\epsneigh(\trj(\g))$. 
The complexity of model checking an $\omega$-regular language with a length of $K^{\epsilon} + p(\g)$ for a given LTL (MLTL) formula $\vp'$ is $\bigO{2^{\bigO{|\vp'|}} \cdot (K^{\epsilon} + p(\g))}$~\cite{HandbookMC18,baier2008principles}. 
Thus the two checks at lines~\ref{algo:line:NBAintersectionEmptiness} and~\ref{algo:line:NBAinclusion} have a total complexity of $\bigO{(2^{\bigO{|\vp|}} + 2^{\bigO{|\neg\vp|}}) \cdot (K^{\epsilon} + p(\g))}$. Since it is widely accepted in the model checking community that the B\"uchi automata $A_{\vp}$ and $A_{\neg\vp}$ can be obtained with the same complexity $2^{\bigO{|\vp|}}$ (cf.\@ the automaton construction given in~\cite[Theorem~5.37]{baier2008principles} or in~\cite[Section~4.6.1]{HandbookMC18}), the term $2^{\bigO{|\vp|}} + 2^{\bigO{|\neg\vp|}}$ becomes $2 \cdot 2^{\bigO{|\vp|}}$. Therefore, the complexity for performing model checking on the $\omega$-regular language $\epsneigh(\trj(\g))$ is $\bigO{2^{\bigO{|\vp|}} \cdot (K^{\epsilon} + p(\g))}$.

In summary, the overall complexity of Algorithm~\ref{ModelCheck} is $\bigO{2^{\bigO{|\vp|}} \cdot (K^{\epsilon} + p(\g))+d^8}$, which concludes the proof.
\qed

\end{proof}

\end{document}